\documentclass{article}

\usepackage{amsmath,amsfonts,amssymb,amsthm}
\usepackage{graphicx,subfig}
\usepackage{tabularx}
\usepackage[utf8]{inputenc}
\usepackage{authblk}
\usepackage{xcolor}
\usepackage{verbatim}
\usepackage[title]{appendix}
\usepackage{paralist}
\usepackage{url}
\usepackage{mathtools}
\newtheorem{theorem}{Theorem}
\newtheorem{lemma}{Lemma}
\newtheorem{assumption}{Assumption}
\newtheorem{corollary}{Corollary}
\theoremstyle{definition}
\newtheorem{remark}{Remark}

\usepackage{enumerate}

\makeatletter
\def\blfootnote{\gdef\@thefnmark{}\@footnotetext}
\makeatother

\renewcommand{\le}{\leqslant}
\renewcommand{\ge}{\geqslant}
\renewcommand{\leq}{\leqslant}
\renewcommand{\geq}{\geqslant}
\renewcommand{\emptyset}{\varnothing}

\newcommand{\ca}{\mathcal{A}}
\newcommand{\ck}{\mathcal{K}}

\newcommand{\tmod}{\ \mathsf{mod}\ }

\newcommand{\real}{\mathbb{R}}
\newcommand{\ints}{\mathbb{Z}}
\newcommand{\natu}{\mathbb{N}}

\newcommand{\bsa}{\boldsymbol{a}}
\newcommand{\bsb}{\boldsymbol{b}}

\newcommand{\bsk}{\boldsymbol{k}}

\newcommand{\bsx}{\boldsymbol{x}}

\newcommand{\bszero}{\boldsymbol{0}}
\newcommand{\bsone}{\boldsymbol{1}}
\newcommand{\bsell}{\boldsymbol{\ell}}
\newcommand{\bsalpha}{\boldsymbol{\alpha}}

\newcommand{\rd}{\,\mathrm{d}}

\newcommand{\e}{\mathbb{E}}
\newcommand{\var}{\mathrm{Var}}

\newcommand{\giv}{\!\mid\!}

\newcommand{\tran}{\mathsf{T}}

\newcommand{\walk}{\mathrm{wal}_k}
\newcommand{\walbsk}{\mathrm{wal}_{\bsk}}

\newcommand{\walkappa}[1]{{\mathrm{wal}_{#1}}}

\newcommand{\supp}{\boldsymbol{s}}
\newcommand{\tu}{\Tilde{t}_u}
\newcommand{\tj}{\Tilde{t}_j}
\newcommand{\tildeK}{\Tilde{\mathcal{K}}}

\begin{document}
\title{Dimension-independent convergence rates of randomized nets using median-of-means}
\author[1,2]{Zexin Pan}
  \date{}
\affil[1]{Johann Radon Institute for Computational and Applied Mathematics,
  \"OAW, Altenbergerstrasse~69, 4040~Linz, Austria.}
\affil[2]{Institute of Fundamental and Transdisciplinary Research, Zhejiang University, 866 Yuhangtang Road, Xihu District, Hangzhou, Zhejiang Province, 310058, China.}
\renewcommand\Affilfont{\footnotesize}
\blfootnote{\noindent Email addresses: \url{zep002@zju.edu.cn}}
\maketitle

%    Abstract is required.
\begin{abstract}
Recent advances in quasi-Monte Carlo integration demonstrate that the median of linearly scrambled digital net estimators achieves near-optimal convergence rates for high-dimensional integrals without requiring a priori knowledge of the integrand's smoothness. Building on this framework, we prove that the median estimator attains dimension-independent convergence, a property known as strong tractability in complexity theory, under tractability conditions characterized by low effective dimensionality. Using a probabilistic, integrand-specific error criterion, our analysis establishes both faster and dimension-independent convergence under weaker assumptions than previously possible in the worst-case setting.
\end{abstract}

\section{Introduction}

The quasi-Monte Carlo (QMC) method is a powerful alternative to conventional Monte Carlo (MC) techniques for high-dimensional integration. While both methods approximate integrals by averaging $n$ function evaluations, QMC replaces MC’s pseudorandom sampling with low-discrepancy point sets engineered to uniformly explore the sampling space. This structured approach allows QMC to avoid the exponential cost growth with dimension of classical quadrature rules while achieving faster error convergence than MC. These advantages have made QMC particularly valuable in fields like computational finance \cite{joy1996quasi,l2009quasi} and physics \cite{kuo:nuye:2016}, where high-dimensional integrals are pervasive.

Building on these strengths, we focus on randomized quasi-Monte Carlo (RQMC) methods based on linearly scrambled base-2 digital nets. Digital nets are a widely studied QMC construction that achieves low-discrepancy sampling through deterministic point sets \cite{dick:pill:2010}. The random linear scrambling \cite{mato:1998:2}, a computationally efficient alternative to Owen’s scrambling \cite{rtms}, preserves the structure of digital nets while introducing controlled randomness. This randomization enhances the root mean squared error convergence rate to nearly $O(n^{-3/2})$ under smoothness conditions and enables statistical inference. However, traditional approaches that estimate integrals by averaging independent QMC replicates face a critical limitation: outliers inherent to linearly scrambled digital nets can prevent these methods from attaining faster convergence rates, even when the integrand satisfies stronger smoothness assumptions \cite{superpolyone}.

To address this limitation, we adopt the median trick, a methodological innovation that has inspired many recent advances in high-dimensional integration \cite{chen2025randomintegrationalgorithmhighdimensional, goda2024simpleuniversalalgorithmhighdimensional,goda:lecu:2022,Goda2024, ye2025medianqmcmethodunbounded}. The median of linearly scrambled digital net estimators was first analyzed in \cite{superpolyone}, which demonstrated a convergence rate of $O(n^{-c\log(n)})$ for one-dimensional analytic functions on $[0,1]$, where $c<3\log(2)/\pi^2\approx 0.21$. Subsequent work in \cite{superpolymulti} extended this result to higher dimensions. For finitely differentiable integrands, \cite{pan2024automatic} further showed that the median estimator achieves near-optimal asymptotic convergence rates across diverse function spaces, requiring neither a priori smoothness parameters nor pre-optimized digital nets.

A critical gap in these prior studies is the role of dimensionality. We focus on the regime where the dimension grows unboundedly while the function norm remains bounded—a common scenario in path integral approximations \cite{kuo:sloa:wasi:wozn:2010}. Our analysis demonstrates that the median estimator attains dimension-independent convergence, a property termed strong tractability in information-based complexity \cite{nova:wozn:2008}, under assumptions on the decay of the relative variation $\gamma_u$ of each ANOVA component (formally defined in Section~\ref{sec:anova}). This notion of relative variation aligns with the effective dimension framework introduced in \cite{cafmowen}, which quantifies how high-dimensional integrands can concentrate their variation in low-order interactions. The connection to weighted function spaces \cite{sloa:wozn:1998} is discussed in Remark~\ref{rmk:weighted}, further bridging our tractability conditions to classical QMC theory.

Before presenting our results, we compare our contributions to the tractability analysis of median digital nets in \cite{Goda2024}. A key distinction lies in the error criterion: their analysis employs the worst-case integration error over the unit ball in a weighted Sobolev space, whereas we demonstrate that median estimators achieve near-optimal convergence rates on a fixed integrand with high probability. Our probabilistic approach is justified by the exponential decay in the likelihood of large errors as the number of replicates increases, aligning with the $(\varepsilon,\delta)$-approximation framework studied in \cite{kunsch2019solvable, kunsch2019optimal}.

Within this framework, we prove, informally stated, the following result:
\begin{theorem}
    Given $\eta>0$, $k\in \natu$, $r\in \natu$ and $f\in C^{(k,\dots,k)}([0,1]^s)$, it holds with probability $1-2^{-r}$ that the median estimator $\hat{\mu}^{(r)}_{\infty}$ from Subsection~\ref{subsec:mse} satisfies
    $$\left|\hat{\mu}^{(r)}_{\infty}-\int_{[0,1]^s}f(\bsx)\rd\bsx\right|\leq C_{\eta,k,f} n^{-k-1/2+\eta}.$$
   Here, $C^{(k,\dots,k)}([0,1]^s)$ denotes functions with continuous dominating mixed partial derivatives of order $k$ (see Subsection~\ref{subsec:smoothness}).  The constant $C_{\eta,k,f}$ is independent of $s$ under the random linear scrambling if the relative variation $\gamma_u$ satisfies the tractability conditions in Theorem~\ref{thm:tractability}.
\end{theorem}

These nearly $O(n^{-k-1/2})$ integrand-specific error rates complement the nearly $O(n^{-k})$ worst-case rates established in \cite{Goda2024}. Furthermore, Theorem~\ref{thm:tractability} requires strictly weaker assumptions than those assumed in \cite{Goda2024}  for the random linear scrambling, demonstrating that dimension-independent convergence is more easily achieved under our error criterion.

As a final remark, we note that while several algorithms match lower complexity bounds for integration in specific spaces—such as higher-order digital nets \cite{goda2016explicit,goda2017optimal,goda2018optimal}, Frolov lattice rules \cite{frolov1976upper,krieg2017universal,ullrich2016role} and Smolyak algorithms \cite{gnewuch2020explicit, sickel2006smolyak, smol:1963}—our method offers distinct practical advantages. It requires no smoothness parameter specification, yet automatically matches theoretical lower bounds up to an arbitrarily small $\eta > 0$ in the convergence rate, making it both efficient and immune to parameter misspecification. Combined with its straightforward implementation and the dimension-independent convergence established in this paper, the median estimator presents a highly practical and robust alternative for applied problems. Furthermore, recent work has successfully extended median RQMC to $L^2$-approximation \cite{pan2025universal,pan2025l_2}, demonstrating its broader utility beyond numerical integration.

The remainder of this paper is structured as follows. Section~\ref{sec:back} reviews essential background on digital nets and establishes notation. Subsection~\ref{subsec:mse} in particular details how the median trick enables probabilistic error guarantees and accelerates convergence rates. Section~\ref{sec:anova} introduces ANOVA decomposition and formally defines the relative variation $\gamma_u$, a key measure of effective dimensionality. Section~\ref{sec:main} provides a foundational analysis of the median estimator’s error bounds and tractability under the completely random design, a simplified randomization scheme that isolates core theoretical insights. Section~\ref{sec:general} generalizes these results to a broader class of randomization schemes, with the random linear scrambling serving as a key example. Section~\ref{sec:exper} empirically validates our theoretical findings through simulations on high-dimensional integrands. Finally, Section~\ref{sec:disc} concludes the paper with a discussion of future research directions.

\section{Background and notation}\label{sec:back}

Let $\natu=\{1,2,3,\dots\}$ denote the natural numbers and $\natu_0=\natu\cup\{0\}$. For $s$-dimensional indices, we define $\natu^s_* = \natu_0^s\setminus\{\bszero\}$ to exclude the zero vector. We define $\ints_{\le\ell}=\{0,1,\dots,\ell\}$ for $\ell\in \natu_0$ and $\ints_{<\ell}=\{0,1,\dots,\ell-1\}$ for $\ell\in \natu$. The dimension of the integration domain is denoted by $s$, and we write $1{:}s=\{1,2,\dots,s\}$ for $s\in \natu$. For a vector $\bsx\in\real^s$ and a subset $u\subseteq 1{:}s$, $\bsx_u$  denotes the subvector of $\bsx$ indexed by $u$, while $\bsx_{u^c}$ denotes the subvector indexed by $1{:}s\setminus u$. The indicator function $\bsone\{\mathcal{A}\}$ equals $1$ if event $\mathcal{A}$ occurs and $0$ otherwise. For $\bsell=(\ell_1,\dots,\ell_s)\in\natu^s_0$, $\Vert \bsell\Vert_1=\sum_{j=1}^s \ell_j$. The cardinality of a set $K$ is $|K|$. For positive sequences $a_n$ and $b_n$, $a_n=O(b_n)$ denotes asymptotic dominance: $a_n\le Cb_n$ for some constant $C<\infty$ and all sufficiently large $n\in\natu$. 

We assume the integrand $f:[0,1]^s\to\mathbb{R}$ is continuous, denoted as $f\in C([0,1]^s)$, with $\Vert f\Vert_\infty=\max_{\bsx\in [0,1]^s}|f(\bsx)|$. The  $L^p$-norm of $f$ is $\Vert f\Vert_{L^p}=\big(\int_{[0,1]^s}|f(\bsx)|^p\rd\bsx\big)^{1/p}$, and $L^p([0,1]^s)$ denotes the space of functions with finite $L^p$-norm.

 Quasi-Monte Carlo (QMC) methods approximate
\begin{equation*}
    \mu=\int_{[0,1]^s}f(\bsx)\rd\bsx
\quad
\text{by}
\quad
\hat\mu=\frac1n\sum_{i=0}^{n-1}f(\bsx_i)
\end{equation*}
for points $\bsx_i\in[0,1]^s$. In this paper, the points $\bsx_i$ are constructed by base-2 digital nets defined in the next subsection.

\subsection{Digital net construction and randomization}\label{subsec:nets}

For $m\in \natu$  and $i\in \ints_{<2^m}$, we express $i$ in its binary form $i=\sum_{\ell=1}^{m}i_\ell 2^{\ell-1}$, encoded as the column vector $\vec{i}= (i_1,\dots,i_m)^\tran\in\{0,1\}^m$. Analogously, for $a\in[0,1)$, we approximate its infinite binary expansion $a=\sum_{\ell=1}^\infty a_\ell 2^{-\ell}$ by truncating to $E\in \natu$ digits, yielding $\vec{a}=(a_1,\dots,a_E)^\tran\in\{0,1\}^E$. For numbers with dual representations (such as dyadic rationals), we select the finite expansion ending in zeros.

A base-2 digital net in $[0,1]^s$ is determined by  $s$ matrices $C_j\in\{0,1\}^{E\times m}$. The deterministic points $\bsx_i=(x_{i1},\dots,x_{is})$ are constructed via:
\begin{align*}
\vec{x}_{ij} = C_j\vec{i} \ \tmod 2 \text{ for } i\in \ints_{<2^m}, j\in1{:}s,
\end{align*}
where $\vec{x}_{ij}\in \{0,1\}^E$ maps to $x_{ij}\in [0,1)$ by appending zeros to the truncated $E$-digit representation. Classical constructions typically assume $E\leq m$.

To introduce randomness, we modify the construction as:
\begin{align}\label{eqn:xequalMCiplusD}
\vec{x}_{ij} = C_j\vec{i} + \vec{D}_j\ \tmod 2,
\end{align}
where $C_j\in \{0,1\}^{E\times m}$ and $\vec{D}_j\in\{0,1\}^E$ are random with precision $E\ge m$. The vector $\vec{D}_j$ is called the \textbf{digital shift} and comprises independent Bernoulli$(0.5)$ entries. A common way to randomize $C_j$ is the \textbf{random linear scrambling} \cite{mato:1998:2}:
$$C_j=M_j\mathcal{C}_j\ \tmod 2,$$
where $M_j\in \{0,1\}^{E\times m}$ is a random lower-triangular matrix (diagonal entries fixed at $1$, subdiagonal entries independently drawn from Bernoulli$(0.5)$), and $\mathcal{C}_j \in \{0,1\}^{m\times m}$ is a predetermined generating matrix optimized to improve the discrepancy of resulting points $\bsx_i$ (see \cite[Chapter 4.4]{dick:pill:2010} for details).

While effective, the reliance of random linear scrambling on carefully designed generating matrices complicates theoretical analysis. We address this in Section~\ref{sec:general}. For initial insights, Section~\ref{sec:main} adopts the \textbf{complete random design} \cite{pan2024automatic}, where all entries of $C_j$ are drawn independently from Bernoulli$(0.5)$. This approach retains the asymptotic convergence rate of random linear scrambling while circumventing the need for pre-optimized generator matrices.

Let $\bsx_i[E]$ denote points generated by equation~\eqref{eqn:xequalMCiplusD} with $E$-bit precision. The QMC estimator becomes:
\begin{equation}\label{eqn:muEdef}
   \hat\mu_E = \frac1n\sum_{i=0}^{n-1}f(\bsx_i)\quad\text{for $\bsx_i=\bsx_i[E]$}. 
\end{equation}
We conveniently assume $E=\infty$ and focus our theoretical analysis on  $\hat{\mu}_\infty$. In practice, $E$ is constrained by computational precision.  Lemma 1 of \cite{superpolymulti} shows $|\hat\mu_E-\hat\mu_\infty|\leq \omega_f(\sqrt{s}2^{-E})$, where $\omega_f$ is the modulus of continuity of $f$. Thus, $\hat\mu_E$ closely approximates $\hat\mu_\infty$ when $ \omega_f(\sqrt{s}2^{-E})\ll |\hat{\mu}_E-\mu|$.

\subsection{Walsh functions and error analysis for base-2 digital nets}

Walsh functions form the natural orthonormal basis for analyzing base-2 digital nets. For $k\in\natu_0$ and $x\in[0,1)$, the univariate
Walsh function $\walk(x)$ is defined as
$$
\walk(x) = (-1)^{\vec{k}^\tran\vec{x}},
$$
where $\vec{k}\in \{0,1\}^\infty$ and $\vec{x}\in \{0,1\}^\infty$ are the binary expansions of $k$ and $x$, respectively. Since $\vec{k}$ has finitely many nonzero digits, finite-precision truncations of $\vec{x}$ suffice for computation.

For $s$-dimensional functions, the multivariate Walsh function $\walbsk:[0,1)^s\to\{-1,1\}$ is the tensor product of univariate Walsh functions:
$$
\walbsk(\bsx) =\prod_{j=1}^s\mathrm{wal}_{k_j}(x_j)
=(-1)^{\sum_{j=1}^s\vec{k}_j^\tran \vec{x}_j},
$$
where $\bsk=(k_1,\dots,k_s)\in \natu^s_0$. These functions constitute a complete orthonormal basis for $L^2([0,1]^s)$ \cite{dick:pill:2010}, enabling the Walsh series decomposition:
\begin{align}\label{eqn:Walshdecomposition}
f(\bsx) = \sum_{\bsk\in\natu_0^s}\hat f(\bsk) \walbsk(\bsx),\quad\text{where}\quad
\hat f(\bsk) = \int_{[0,1]^s}f(\bsx)\walbsk(\bsx)\rd\bsx.
\end{align}
The equality holds in the $L^2$-sense. Using this decomposition, \cite[Theorem 1]{pan2024automatic} establishes a key error structure for QMC estimators:
\begin{lemma}\label{lem:decomp}
For $f\in C([0,1]^s)$, the error of $\hat{\mu}_\infty$ defined by equation~\eqref{eqn:muEdef} satisfies
\begin{equation}\label{eqn:errordecomposition}
    \hat{\mu}_{\infty}-\mu=\sum_{\bsk\in \natu_*^s}Z(\bsk)S(\bsk)\hat{f}(\bsk),
\end{equation}
where
$$Z(\bsk)=\bsone\Big\{\sum_{j=1}^s \vec{k}_j^\tran  C_j=\bszero \tmod 2\Big\} \quad\text{and}\quad S(\bsk) = (-1)^{\sum_{j=1}^s\vec{k}_j^\tran \vec{D}_j}.$$
\end{lemma}
Here, $Z(\bsk)$ is an indicator for failure to integrate the Walsh mode $\walbsk$ exactly, and $S(\bsk)$ encodes the random sign induced by the digital shifts $\vec{D}_j$.

\subsection{Probabilistic error guarantee of the median estimator}\label{subsec:mse}
Given $r\in \natu$, we consider $2r-1$ independently generated replicates of 
$\hat{\mu}_\infty$ ordered from the smallest to the largest:
$$\hat{\mu}^{(1)}_{\infty}\leq \dots \leq \hat{\mu}^{(r)}_{\infty}\leq \dots \leq \hat{\mu}^{(2r-1)}_{\infty}.$$
A key property of the sample median $\hat{\mu}^{(r)}_{\infty}$ is the following probabilistic error bound:

\begin{lemma}\label{lem:medianMSE}
    Let $K\subseteq \natu^s_*$ and $\mathcal{A}$ be the event $\{Z(\bsk)=1 \text{ for some } \bsk\in K\}$. If $\Pr(\ca)\leq \delta<1/8$, then
    \begin{equation*}
    \Pr(|\hat{\mu}^{(r)}_{\infty}-\mu|^2>\varepsilon_K/\delta)\leq (8\delta)^r \quad \text{for} \quad \varepsilon_K=\sum_{\bsk\in\natu^s_*\setminus K} \Pr(Z(\bsk)=1) |\hat{f}(\bsk)|^2.
\end{equation*}
\end{lemma}
\begin{proof}
     First note that the random signs $S(\bsk)$ in equation~\eqref{eqn:errordecomposition} are  zero-mean and pairwise-independent for distinct $\bsk\in \natu_*^s$ (see \cite[Lemma 4]{superpolyone} for the one-dimensional case). Since $\mathcal{A}$ is independent of $S(\bsk)$, we have $\e(\hat{\mu}_\infty\giv \ca^c)=\mu$ and
     $$\var(\hat{\mu}_\infty\mid \ca^c)=\sum_{\bsk\in\natu^s_*\setminus K} \Pr(Z(\bsk)=1\mid \ca^c) |\hat{f}(\bsk)|^2\leq \frac{\varepsilon_K}{\Pr(\mathcal{A}^c)}.$$
     Applying Chebyshev’s inequality yields
     \begin{align*}
     \Pr(|\hat{\mu}_{\infty}-\mu|^2>\varepsilon_K/\delta)\leq & \Pr(\mathcal{A})+\Pr(\mathcal{A}^c)\Pr(|\hat{\mu}_{\infty}-\mu|^2>\varepsilon_K/\delta\mid \mathcal{A}^c)\leq 2\delta.
     \end{align*}
     Finally, $|\hat{\mu}^{(r)}_{\infty}-\mu|> \varepsilon_K/\delta$ only if the event $|\hat{\mu}_{\infty}-\mu|^2>\varepsilon_K/\delta$ occurs for at least $r$ out of the $2r-1$ independent replicates. A union bound then implies
     $$\Pr(|\hat{\mu}^{(r)}_{\infty}-\mu|^2>\varepsilon_K/\delta)\leq {2r-1\choose r}(2\delta)^{r}\leq 2^{2r-1}(2\delta)^{r}\leq (8\delta)^r.\qedhere$$
\end{proof}

This lemma justifies the use of the sample median $ \hat{\mu}^{(r)}_{\infty}$ as an estimator for $\mu$. For
 $f\in C([0,1]^s)$, the mean squared error satisfies
$$\e|\hat{\mu}^{(r)}_{\infty}-\mu|^2\leq \varepsilon_K/\delta+(8\delta)^r\Vert f\Vert^2_\infty.$$
Because the second term decays exponentially in $r$, it can be made negligible relative to the first term by choosing $r$ sufficiently large.

The main challenge is therefore to select a set $K$ so that $\varepsilon_K$ converges at the desired rate. We address this by constructing $K$ adaptively  based on the decay rate of the Walsh coefficients $\hat{f}(\bsk)$, which encodes the integrand’s smoothness and effective dimensionality. Notably, the median estimator $\hat{\mu}^{(r)}_{\infty}$  does not require explicit knowledge of $K$, making it applicable without prior information.

\subsection{Walsh coefficients and function smoothness}\label{subsec:smoothness}
To analyze the decay of Walsh coefficients, we introduce the following notation. For $k = \sum_{\ell=1}^\infty a_\ell 2^{\ell-1}$, let $\kappa = \{\ell\in \natu\mid a_\ell=1\}$ denote the positions of nonzero bits in the binary expansion of $k$. We use $k$ and $\kappa$ interchangeably because each uniquely determines the other. In this framework, $|\kappa|$ equals the number of nonzero bits of $k$, and the Walsh function simplifies to:
$$\walkappa{\kappa}(x)=\prod_{\ell\in \kappa}(-1)^{\vec{x}(\ell)},$$
where $\vec{x}(\ell)$ denotes the $\ell$-th digit of $x$  in its binary expansion.

For a finite subset $\kappa\subseteq \natu$, we define
\begin{itemize}
    \item $\lceil\kappa\rceil_{q}$: The  $q$-th largest element of $\kappa$ (or $0$ if $|\kappa|<q$).
    \item $\lceil\kappa\rceil_{1{:}q}$: The set of the $q$ largest elements of $\kappa$ (or $\emptyset$ if $q=0$).
\end{itemize}

The following lemma connects the Walsh coefficients $\hat{f}(\bsk)$ to the smoothness of $f$. For $\bsalpha=(\alpha_1,\dots,\alpha_s)\in \natu^s_0$, we define the mixed partial derivative:
$$f^{(\bsalpha)}=f^{(\alpha_1,\dots,\alpha_s)}=\frac{\partial^{\Vert\bsalpha\Vert_1} f}{\partial x_1^{\alpha_1}\cdots\partial x_s^{\alpha_s}},$$
with the convention $f^{(\bszero)}=f$.

\begin{lemma}\label{lem:exactfk}
Let $\bsalpha\in \natu^s_0$ and $f^{(\bsalpha)}\in C([0,1]^s)$. If $ |\kappa_j|\geq \alpha_j$ for all $j\in 1{:}s$, then
\begin{equation}\label{eqn:exactfk}
\hat f(\bsk)=(-1)^{\Vert\bsalpha\Vert_1}\int_{[0,1]^s}f^{(\bsalpha)}(\bsx)\prod_{j=1}^s \walkappa{\kappa_j\setminus \lceil\kappa_j\rceil_{1{:}\alpha_j}}(x_j)W_{\lceil\kappa_j\rceil_{1{:}\alpha_j}}(x_j)\rd \bsx,
\end{equation} 
where $W_\kappa :[0,1]\to\mathbb{R}$ for $\kappa\subseteq \natu$ is defined recursively by $W_\emptyset(x)=1$ and 
\begin{equation*}
  W_{\kappa}(x)=\int_{[0,1]}(-1)^{\vec{x}(\lfloor\kappa\rfloor)}W_{\kappa\setminus \lfloor\kappa\rfloor}(x)\rd x
\end{equation*}
with $\lfloor\kappa\rfloor$ denoting the smallest element of $\kappa$. Furthermore, for nonempty $\kappa$, $W_\kappa(x)$ is continuous, nonnegative, periodic with period $2^{-\lfloor\kappa\rfloor+1}$ and
\begin{equation}\label{eqn:Wint}
\int_{[0,1]}W_\kappa(x)\rd x=\prod_{\ell\in\kappa}2^{-\ell-1}\quad \text{and} \quad \max_{x\in [0,1]} W_\kappa(x)= 2\prod_{\ell\in\kappa} 2^{-\ell-1}.
\end{equation}
\end{lemma}
\begin{proof}
Equation~\eqref{eqn:exactfk} comes from \cite[Theorem 2.5]{SUZUKI20161}, while properties of $W_\kappa(x)$ are proven in \cite[Section 3]{SUZUKI20161}.
\end{proof}

For $\alpha\in \natu_0$, we say $f\in C^{(\alpha,\dots,\alpha)}([0,1]^s)$ if the mixed partial derivative $f^{(\bsalpha)}$ exists and is continuous on $[0,1]^s$ for all $\bsalpha\in \ints^s_{\le \alpha}$. In this work, we assume $f\in C^{(\alpha,\dots,\alpha)}([0,1]^s)$ in order to directly apply Lemma~\ref{lem:exactfk}. This regularity condition can be relaxed to the existence of weak derivatives $f^{(\bsalpha)}\in L^1([0,1]^s)$ satisfying the integral identity in equation~\eqref{eqn:exactfk}, as discussed in \cite[Remark 1]{pan2024automatic}.

\section{ANOVA decomposition and effective dimensionality}\label{sec:anova}

To characterize the effective dimensionality of integrands $f\in L^2([0,1]^s)$, we employ the ANOVA decomposition \cite{meandim}:
$$f=\sum_{u\subseteq 1{:}s} f_u,$$
where each component $f_u$ depends only on $\bsx_u$ and satisfies
$$\int_{0}^1 f_u(\bsx)\rd x_j=0 \quad\forall j\in u.$$
This decomposition is unique in the $L^2$-sense, with $f_\emptyset=\mu$. For $f\in C([0,1]^s)$, we enforce $f_u\in C([0,1]^s)$ via the recursive definition:
\begin{equation}\label{eqn:fudef}
    f_u(\bsx)=\int_{[0,1]^{s-|u|}}f(\bsx)\rd \bsx_{u^c}-\sum_{\substack{v\subseteq u\\v\neq u}}f_v(\bsx).
\end{equation}
This construction preserves smoothness: $f_u\in C^{(\alpha,\dots,\alpha)}([0,1]^s)$ for all $u\subseteq 1{:}s$ if $f\in C^{(\alpha,\dots,\alpha)}([0,1]^s)$. Here, $f_u$ captures the variation due to interactions among variables $\bsx_u$.  By an abuse of notation, we write $f_u(\bsx)=f_u(\bsx_u)$, where the latter $f_u:[0,1]^{|u|}\to \mathbb{R}$ is a $|u|$-dimensional function.

Let $\supp(\bsk)=\{j\in 1{:}s\mid k_j\neq 0\}$ denote the support of $\bsk=(k_1,\dots,k_s)$. From the Walsh decomposition~\eqref{eqn:Walshdecomposition}, we know that
\begin{equation}\label{eqn:fuWalsh}
    f_u(\bsx)=\sum_{\substack{\bsk\in \natu^s_0\\\supp(\bsk)=u}}\hat f(\bsk) \walbsk(\bsx),
\end{equation}
because $\walbsk(\bsx)$ depends only on $\{x_j\mid j\in \supp(\bsk)\}$ and
$$\int_{0}^1 \walbsk(\bsx) \rd x_j=\Big(\int_{0}^1 \walkappa{\kappa_j}(x_j)\rd x_j\Big)\prod_{j'\in \supp(\bsk)\setminus\{j\}} \walkappa{\kappa_{j'}}(x_{j'}) =0$$
if $j\in \supp(\bsk)$. Thus, $\hat f(\bsk)=\hat f_u(\bsk)$ if $\supp(\bsk)=u$. 

To quantify the variation of ANOVA components  $f_u$, we adopt the fractional Vitali variation with $\mathfrak{p}=2$ from \cite{dick:2008}, which generalizes the usual Sobolev norms. For a function $f:[0,1]^{d}\to \mathbb{R}$ and a subcube $J=\prod_{j=1}^d [a_j,b_j)$, we define the alternating sum over $J$:
\begin{equation*}
\Delta(f,J)=\sum_{v\subseteq 1{:}d}(-1)^{d-|v|}f\Big((\bsa_v,\bsb_{v^c})\Big),
\end{equation*}
where $(\bsa_v,\bsb_{v^c})$ denotes the vertex of $J$ with $x_j=a_j$ if $j\in v$ and $x_j=b_j$ otherwise. For $\lambda \in (0,1]$, the fractional Vitali variation of order $\lambda$ is:
\begin{equation}\label{eqn:variationdef}
    V^{(d)}_\lambda(f)=\sup_{\mathcal{P}}\Bigl(\sum_{J\in \mathcal{P}}\mathrm{Vol}(J)\Bigl|\frac{\Delta(f,J)}{\mathrm{Vol}(J)^\lambda}\Bigr|^2\Bigr)^{1/2}
\end{equation}
where the supremum is over all partitions $\mathcal{P}$ of $[0,1)^d$ into subcubes and $\mathrm{Vol}(J)$ is the volume of subcube $J$. When $d=0$ and $f\in \mathbb{R}$, $V^{(0)}_\lambda(f)=|f|$ by convention. 

For $f$ with weak mixed derivative $f^{(1,\dots,1)}$, we have the inequality
\begin{equation}\label{eqn:V1fL2}
 V^{(d)}_1(f)\leq \Vert f^{(1,\dots,1)}\Vert_{L^2},  
\end{equation}
with equality attained when $f\in C^{(1,\dots,1)}([0,1]^s)$ \cite[equation (3.4)]{dick:2008}.
More generally, \cite{LIU2025101935} shows $V^{(d)}_\lambda(f)\leq \Vert f^{(1,\dots,1)}\Vert_{L^p}$ for $p=2/(3-2\lambda),\lambda\in [1/2,1]$. Discontinuous $f$ can also have finite $V^{(d)}_\lambda(f)$ (see \cite[Section 4.1]{pan2024automatic} for an example).

Next, for $v\subseteq 1{:}d$, we define the partial integral operator $I_v:L^1([0,1]^d)\to L^1([0,1]^{|v|})$ as:
\begin{equation*}
    I_v(f)(\bsx_v)=\int_{[0,1]^{d-|v|}}f(\bsx_{v^c};\bsx_v)\rd\bsx_{v^c},
\end{equation*}
where $f(\bsx_{v^c};\bsx_v)$ denotes $f$ with fixed $\bsx_v$. Let  $D_{v^c}$ be the class of functions:
\begin{equation}\label{eqn:Du}
    D_{v^c}=\Big\{\rho_{v^c}\Big\vert \rho_{v^c}(\bsx)=\prod_{j\in v^c}\rho_j(x_j), \rho_j\in C([0,1]), \rho_j\geq 0, \Vert \rho_j\Vert_\infty\leq 1\Big\}.
\end{equation}
The partial fractional Vitali variation of order $\lambda$ with respect to $v\subseteq 1{:}d$ is:
\begin{equation}\label{eqn:Vvbound}
   V^{v}_\lambda(f)=\sup_{\rho_{v^c}\in D_{v^c}} V^{(|v|)}_\lambda \Big(I_v(f\rho_{v^c})\Big). 
\end{equation}
By convention, $V^{1{:}d}_\lambda(f)=V^{(d)}_\lambda(f)$. From \cite[Section 3.2]{pan2024automatic}, 
\begin{equation*}
    \Bigl(V^{v}_\lambda(f)\Bigl)^2\leq \int_{[0,1]^{d-|v|}}\Bigl(V^{(|v|)}_\lambda\Bigl( f(\text{ }\cdot\text{ };\bsx_{v^c})\Bigr)\Bigr)^2\rd\bsx_{v^c},
\end{equation*}
which, together with equation~\eqref{eqn:V1fL2}, implies
\begin{equation}\label{eqn:V1vfL2}
  V^{v}_1(f)\leq \Vert f^{(\bsone_{v},\bszero_{v^c})}\Vert_{L^2}  
\end{equation}
if the weak derivative $f^{(\bsone_{v},\bszero_{v^c})}$ exists.
Here $(\bsone_{v},\bszero_{v^c})$ denotes the vector whose $j$-th entry equals $1$ if $j\in v$ and equals $0$ otherwise.

The following lemma connects the Walsh coefficients $\hat f(\bsk)=\hat f_u(\bsk)$ (where $\supp(\bsk)=u$) to the fractional Vitali variation of $f_u$:
\begin{lemma}\label{lem:fkbound}
Let $\alpha\in \natu_0$ and $\lambda \in (0,1]$. Suppose $V^{v}_\lambda (f^{(\bsalpha_u)}_u)<\infty$ for nonempty $u\subseteq 1{:}s$, $\bsalpha_u=(\alpha_j,j\in u)\in \ints_{\le \alpha}^{|u|}$ and a subset $v\subseteq \{j\mid \alpha_j=\alpha\} $. Then for $\bsk_u=(k_j,j\in u )\in \natu^{|u|}$ satisfying $|\kappa_j|=\alpha+1$ for $j\in v$ and $|\kappa_j|=\alpha_j$ for $j\in u\setminus v$,
\begin{equation*}
    \sum_{\bsk'\in B(\bsk_u,v)}|\hat{f}_u(\bsk')|^2
\leq 4^{|u|}\Big(V^{v}_\lambda (f^{(\bsalpha_u)}_u)\Big)^2 \prod_{j\in u}4^{(1-\lambda) \lceil\kappa_j\rceil_{\alpha+1}}\prod_{\ell'\in \kappa_j}4^{-\ell'-1},
\end{equation*}
where 
\begin{equation}\label{eqn:Bku}
B(\bsk_u,v)=\{\bsk'\in \natu^s_*\mid \supp(\bsk')=u,\ \lceil\kappa'_j\rceil_{1:(\alpha+1)}=\kappa_j \ \forall j\in v,\ \kappa'_j= \kappa_j \ \forall j \in u\setminus v\}.    
\end{equation}
\end{lemma}
\begin{proof}
See Appendix~\ref{app3}.
\end{proof}

Motivated by Lemma~\ref{lem:fkbound}, we define for $\alpha\in \natu_0, \lambda \in (0, 1]$ and nonempty $u\subseteq 1{:}s$:
\begin{equation}\label{eqn:funorm}
 \Vert f\Vert_{u,\alpha,\lambda}=\sup_{v\subseteq u}\sup_{\substack{\bsalpha_u\in \ints_{\le \alpha}^{|u|}\\ \alpha_j=\alpha \ \forall j\in v,\\ \alpha_j>0 \ \forall j\in u\setminus v  }} V^{v}_\lambda (f^{(\bsalpha_u)}_u),   
\end{equation}
where $f_u$ is the ANOVA component of $f$ defined in equation~\eqref{eqn:fudef}. Notice that we add the restriction $\alpha_j>0$ for $j\in u\setminus v$ because $k_j\in\natu$ implies $\alpha_j=|\kappa_j|>0$ in the statement of Lemma~\ref{lem:fkbound}. In particular, $ \Vert f\Vert_{u,0,\lambda}=V^{u}_\lambda (f_u)$. When $\lambda=1$ and the weak mixed derivative $f^{(\alpha+1,\dots,\alpha+1)}$ exists, equation~\eqref{eqn:V1vfL2} implies the bound
$$\Vert f\Vert_{u,\alpha,1}\leq \sup_{\bsalpha_u\in \{1:(\alpha+1)\}^{|u|}} \Vert f_u^{(\bsalpha_u)}\Vert_{L^2}.$$

 We further set $\Vert f\Vert_{u,\alpha,\lambda}=0$ if $u=\emptyset$, and define the global variation norm:
$$\Vert f\Vert_{\alpha,\lambda}=\sup_{u\subseteq 1{:}s}\Vert f\Vert_{u,\alpha,\lambda}.$$
By Lemma~\ref{lem:fkbound}, $\Vert f\Vert_{\alpha,\lambda}=0$ implies $\hat{f}(\bsk)=0$ for all $\bsk\neq \bszero$, forcing $f$ to be a constant. For non-constant $f$, the \textbf{relative variation} of ANOVA component $f_u$ is:
\begin{equation}\label{eqn:gammau}
    \gamma_u=\gamma_u(f,\alpha,\lambda)=\frac{\Vert f\Vert_{u,\alpha,\lambda}}{\Vert f\Vert_{\alpha,\lambda}},
\end{equation}
which satisfies $\sup_{u\subseteq 1{:}s}\gamma_u=1$ and $\gamma_\emptyset=0$. The relative variation $\gamma_u$ depends on $f,\alpha$ and $\lambda$, but we suppress this dependence in the notation as these variables are fixed throughout our theoretical analysis. Section~\ref{sec:exper} later exemplifies how $\gamma_u$ changes under different choices of $\alpha$.

The decay of $\gamma_u$ encodes effective dimensionality, with $\gamma_u=1$ for all nonempty $u\subseteq 1{:}s$ representing the worst dimensionality. As shown in Theorem~\ref{thm:tractability},  dimension-independent convergence rates emerge under suitable decay conditions on $\gamma_u$.

\begin{remark}\label{rmk:weighted}
    For positive weights $\{\Tilde{\gamma}_u,u\subseteq 1{:}s\}$, we can define the weighted norm:
$$\Vert f\Vert_{\alpha,\lambda,\Tilde{\gamma}_u}=\sup_{u\subseteq 1{:}s}\Tilde{\gamma}_u^{-1}\Vert f\Vert_{u,\alpha,\lambda}.$$
All subsequent results hold equivalently with $\Vert f\Vert_{\alpha,\lambda,\Tilde{\gamma}_u}$ replacing $\Vert f\Vert_{\alpha,\lambda}$ and $\Tilde{\gamma}_u$ replacing $\gamma_u$. However, since $\Vert f\Vert_{\alpha,\lambda}$ and $\gamma_u$ are intrinsic to $f$, while $\Vert f\Vert_{\alpha,\lambda,\Tilde{\gamma}_u}$ requires external weights $\Tilde{\gamma}_u$ unused by our algorithm, we adopt the unweighted framework for simplicity. 

Notably, $\Vert f\Vert_{\alpha,\lambda}$ equals the infimum of $\Vert f\Vert_{\alpha,\lambda,\Tilde{\gamma}_u}$ over all weights $\{\Tilde{\gamma}_u,u\subseteq 1{:}s\}$ with $\sup_{u\subseteq 1{:}s} \Tilde{\gamma}_u=1$. Thus, the median RQMC method implicitly selects optimal weights tailored to $f$, avoiding explicit weight specification.
\end{remark}

\section{Tractability under the completely random design}\label{sec:main}

In this section, we show the median RQMC estimator $\hat{\mu}^{(r)}_\infty$ under the completely random design achieves near $O(2^{-(\alpha+\lambda+1/2)m})$ convergence rates given $\Vert f\Vert_{\alpha,\lambda}<\infty$, with no dimension-dependence if further $\gamma_u$ satisfies equation~\eqref{eqn:summable}.

First,  we define
 $\Vert\kappa\Vert=\sum_{\ell\in \kappa}\ell$ for $\kappa\subseteq\natu$, with $\Vert\emptyset\Vert=0$.
Motivated by Lemma~\ref{lem:fkbound}, we choose $K$ in Lemma~\ref{lem:medianMSE} for given $\alpha\in \natu_0$ and $\lambda\in (0,1]$ to be:
\begin{equation}\label{eqn:setK}
   K=\bigcup_{\substack{u\subseteq 1{:}s\\ u\neq\emptyset}} K_u(T_{u,m}) 
\end{equation}
with thresholds $T_{u,m}$ and
$$K_u(T)=\bigcup_{v\subseteq u}\bigcup_{\bsk_u\in \mathcal{K}_{u,v}(T) }B(\bsk_u,v)$$
for $B(\bsk_u,v)$ defined as in equation~\eqref{eqn:Bku} and
\begin{align*}
    \mathcal{K}_{u,v}(T)=\Big\{(k_j,j\in u )\in \natu^{|u|}\Big\vert \begin{matrix}
    |\kappa_j|=\alpha+1 \ \forall  j\in v,\ |\kappa_j|\leq \alpha \ \forall j\in u\setminus v,\\(\lambda-1)\sum_{j\in u} \lceil\kappa_j\rceil_{\alpha+1}+\sum_{j\in u}\Vert\kappa_j\Vert  \leq T
    \end{matrix} 
     \Big\}.
\end{align*}
We note that $K_u(T)$ implicitly depends on $\alpha$ and $\lambda$. 

The next lemma quantifies how $T$ controls the size of $|\hat{f}_u(\bsk)|^2$ for $\bsk\notin K_u(T)$.

\begin{lemma}\label{lem:fukbound}
Suppose  $\Vert f\Vert_{u,\alpha,\lambda}<\infty$ for $\alpha\in \natu_0,\lambda\in (0,1]$ and nonempty $u\subseteq 1{:}s$. Then for $T\in \mathbb{R}$ and $\theta\in (0,1)$,
    $$\sum_{\bsk\in\natu_*^{s}\setminus K_u(T)} |\hat{f}_u(\bsk)|^2\leq \frac{\Vert f\Vert^2_{u,\alpha,\lambda} }{4^{\theta T}}   C_{\alpha,\lambda,\theta}^{|u|},$$
    where
    \begin{equation}\label{eqn:Cconstant}
    C_{\alpha,\lambda,\theta}=\frac{4^{-\alpha}}{\alpha!(4^{(\alpha+\lambda)(1-\theta)}-1)(4^{1-\theta}-1)^{\alpha}}+\sum_{\alpha'=0}^\alpha \frac{4^{-\alpha'+1}}{(\alpha')!(4^{1-\theta}-1)^{\alpha'}}. 
\end{equation}
\end{lemma}
\begin{proof}
See Appendix~\ref{app4}.
\end{proof}

While larger $T$ implies a tighter bound on $|\hat{f}_u(\bsk)|^2$ with $\bsk\notin K_u(T)$, it also increases the probability that $Z(\bsk)=1$ for some $\bsk\in K_u(T)$. The next lemma provides a bound on the cardinality of $K_u(T)$.
\begin{lemma}\label{lem:KuTbound}
    For $\alpha\in \natu_0,\lambda\in (0,1], T\in\mathbb{R}$ and nonempty $u\subseteq 1{:}s$,
    $$|K_u(T)|\leq 2^{T/(\alpha+\lambda)} \Big(A_{\alpha,\lambda}\max\big(T/(\alpha+\lambda),1\big) +B_{\alpha,\lambda}\Big)^{|u|},$$
    where
\begin{equation}\label{eqn:ABconstant}
    A_{\alpha,\lambda}=\frac{1}{\alpha!(2^{1/(\alpha+\lambda)}-1)^{\alpha}} \quad \text{and}\quad B_{\alpha,\lambda}=\sum_{\alpha'=1}^\alpha \frac{1}{(\alpha')!(2^{1/(\alpha+\lambda)}-1)^{\alpha'}},
\end{equation}
with the sum from $1$ to $\alpha$ interpreted as $0$ when $\alpha=0$.
\end{lemma}
\begin{proof}
See Appendix~\ref{app1}.
\end{proof}

\begin{corollary}\label{cor:eventA}
    For $m\in \natu$ and $\gamma_u$ defined by equation~\eqref{eqn:gammau}, let
    \begin{equation}\label{eqn:Gamma}
     \Gamma_{\alpha,\lambda,\gamma_u}(m)=1+\sum_{\substack{u\subseteq 1{:}s\\ u\neq\emptyset}}\gamma^{\frac{1}{\alpha+\lambda+1/2}}_u (A_{\alpha,\lambda}m+B_{\alpha,\lambda})^{|u|}.   
    \end{equation}
    Then for any $\delta\in (0,1)$, 
    $$\Pr\Big(Z(\bsk)=1 \text{ for some } \bsk\in \bigcup_{\substack{u\subseteq 1{:}s\\ u\neq\emptyset}} K_u(T_{u,m})\Big)\leq \delta$$
    under the completely random design, where
    \begin{equation}\label{eqn:Tum}
        T_{u,m}=(\alpha+\lambda)m-(\alpha+\lambda)\Big(\log_2\big(\Gamma_{\alpha,\lambda,\gamma_u}(m)\big)-\log_2(\delta)-\frac{\log_2(\gamma_u)}{\alpha+\lambda+1/2}\Big).
    \end{equation}
\end{corollary}
\begin{proof}
    Because $\Gamma_{\alpha,\lambda,\gamma_u}(m)\geq 1, \delta<1$ and $\sup_{u\subseteq 1{:}s}\gamma_u=1$, we have $T_{u,m}/(\alpha+\lambda)\leq m$ for all nonempty $u\subseteq 1{:}s$. Lemma~\ref{lem:KuTbound} then implies
    \begin{align*}
        \sum_{\substack{u\subseteq 1{:}s\\u\neq \emptyset} } |K_u(T_{u,m})| \leq &  \sum_{\substack{u\subseteq 1{:}s\\u\neq \emptyset} }2^{T_{u,m}/(\alpha+\lambda)}(A_{\alpha,\lambda}m+B_{\alpha,\lambda})^{|u|} \\
        = & \frac{\delta  2^m }{\Gamma_{\alpha,\lambda,\gamma_u}(m)}\sum_{\substack{u\subseteq 1{:}s\\u\neq \emptyset} }\gamma^{\frac{1}{\alpha+\lambda+1/2}}_u(A_{\alpha,\lambda}m+B_{\alpha,\lambda})^{|u|}\\
        \leq & \delta 2^m.
    \end{align*}
    Under the completely random design, $\Pr(Z(\bsk)=1)=2^{-m}$ for any $\bsk\neq \bszero$, so our conclusion follows by taking a union bound.
\end{proof}

Now we are ready to prove the following error bounds on $\hat{\mu}^{(r)}_{\infty}$:

\begin{theorem}\label{thm:CRD}
Suppose $\Vert f\Vert_{\alpha,\lambda}<\infty$ for $\alpha\in \natu_0$ and $\lambda\in (0,1]$. Then for $m\in \natu, r\in \natu$ and $\theta\in (0,1)$, the completely random design ensures
\begin{equation}\label{eqn:errorbound}
    \Pr\Big(|\hat{\mu}^{(r)}_{\infty}-\mu|^2>\Gamma_{\alpha,\lambda,\gamma_u,\theta}\Vert f\Vert^2_{\alpha,\lambda}\frac{(\Gamma_{\alpha,\lambda,\gamma_u}(m))^{2\theta(\alpha+\lambda)}}{2^{(2\theta(\alpha+\lambda)+1)m}}\Big)\leq  2^{-r},
\end{equation}
where $\gamma_u$ is given by equation~\eqref{eqn:gammau},  $\Gamma_{\alpha,\lambda,\gamma_u}(m)$ is given by equation~\eqref{eqn:Gamma} and 
$$\Gamma_{\alpha,\lambda,\gamma_u,\theta}=16^{2\theta(\alpha+\lambda)+1}\sum_{\substack{u\subseteq 1{:}s\\ u\neq\emptyset}} \gamma^{\frac{1+2(1-\theta)(\alpha+\lambda)}{\alpha+\lambda+1/2}}_u  C_{\alpha,\lambda,\theta}^{|u|},$$
with $C_{\alpha,\lambda,\theta}$ given by equation~\eqref{eqn:Cconstant}.
\end{theorem}
\begin{proof}
Let $\delta=1/16$ in Lemma~\ref{lem:medianMSE} and Corollary~\ref{cor:eventA}. By Lemma~\ref{lem:medianMSE} and $\Pr(Z(\bsk)=1)=2^{-m}$ for $\bsk\neq\bszero$, it suffices to show 
    $$2^{-m+4}\sum_{\bsk\in\natu^s_*\setminus K} |\hat{f}(\bsk)|^2\leq \Gamma_{\alpha,\lambda,\gamma_u,\theta}\Vert f\Vert^2_{\alpha,\lambda}\frac{\Gamma_{\alpha,\lambda,\gamma_u}(m)^{2\theta(\alpha+\lambda)}}{2^{(2\theta(\alpha+\lambda)+1)m}}$$
    for $K$ defined by equation~\eqref{eqn:setK} with $T_{u,m}$ given by equation~\eqref{eqn:Tum}.  From equation~\eqref{eqn:fuWalsh} and Lemma~\ref{lem:fukbound},
\begin{align}\label{eqn:ftofu}
   \sum_{\bsk\in\natu^s_*\setminus K} |\hat{f}(\bsk)|^2
   =\sum_{\substack{u\subseteq 1{:}s\\ u\neq\emptyset}}\sum_{\bsk\in\natu_*^{s}\setminus K_u(T_{u,m})}|\hat{f}_u(\bsk)|^2
   \leq   \sum_{\substack{u\subseteq 1{:}s\\ u\neq\emptyset}} \frac{\Vert f\Vert^2_{u,\alpha,\lambda} }{4^{\theta T_{u,m}}}   C_{\alpha,\lambda,\theta}^{|u|}.
\end{align}
Our conclusion follows by plugging in $\Vert f\Vert^2_{u,\alpha,\lambda} =\gamma^2_u\Vert f\Vert^2_{\alpha,\lambda}$ and equation~\eqref{eqn:Tum}.
\end{proof}

\begin{remark}\label{rmk:asymp}
Given a failure probability tolerance $\bar{\delta}>0$, we can set $r=\lceil\log_2(\bar{\delta}^{-1})\rceil$ to ensure $2^{-r}\leq \bar{\delta}$. For this choice of $r$,
    the median estimator $\hat{\mu}^{(r)}_{\infty}$ asymptotically achieves a convergence rate $O(N^{-\alpha-\lambda-1/2+\eta})$ for any $\eta>0$, where $N=r2^m$ is the total number of function evaluations.
    
    To justify this, observe that $\Gamma_{\alpha,\lambda,\gamma_u}(m)^{2(\alpha+\lambda)}=O(2^{\eta m })$ asymptotically  for any $\eta>0$ because $\Gamma_{\alpha,\lambda,\gamma_u}(m)$ is a polynomial of $m$. Setting $\theta=1-\eta/(2(\alpha+\lambda))$, we derive for all $f$ with $\Vert f\Vert_{\alpha,\lambda}<\infty$ that
    \begin{equation}\label{eqn:erroboundasymp}
    \Pr\Big(|\hat{\mu}^{(r)}_{\infty}-\mu|^2> C_\eta \Vert f\Vert^2_{\alpha,\lambda} 2^{-(2\alpha+2\lambda+1-2\eta)m}\Big)\leq 2^{-r}\leq \bar{\delta},    
    \end{equation}
    where $C_\eta$ is a constant depending on $\alpha,\lambda,\gamma_u$ and $\eta$. This convergence rate is nearly optimal, as it aligns with the lower bound established in \cite{bakh:1959} for one-dimensional $f$ with $\lambda$-Hölder continuous $f^{(\alpha)}$, a subclass of functions with finite $\Vert f\Vert_{\alpha,\lambda}$.
\end{remark}

\begin{remark}\label{rmk:tractability}
    In settings where $s\to\infty$ while $\Vert f\Vert_{\alpha,\lambda}\leq 1$, dimension-independent error bounds hold if there exist weights $\{\gamma_j,j\in \natu\}$ satisfying
    \begin{equation}\label{eqn:summable}
     \gamma_u\leq \prod_{j\in u}\gamma_j\quad \forall u\subseteq \natu \quad \text{ and } \quad C_\infty=\sum_{j=1}^\infty\gamma^{\frac{1}{\alpha+\lambda+1/2}}_j<\infty.  
    \end{equation}
    To see this, we first bound
    \begin{align*}
    \frac{\Gamma_{\alpha,\lambda,\gamma_u,\theta}}{16^{2\theta(\alpha+\lambda)+1}}\leq \sum_{\substack{u\subseteq 1{:}s\\ u\neq\emptyset}} \gamma^{\frac{1}{\alpha+\lambda+1/2}}_u  C_{\alpha,\lambda,\theta}^{|u|}   
    \leq \prod_{j=1}^s (1+ \gamma^{\frac{1}{\alpha+\lambda+1/2}}_j  C_{\alpha,\lambda,\theta})
    \leq \exp(C_{\infty}C_{\alpha,\lambda,\theta}),
    \end{align*}
    where we have used $1+x\leq \exp(x)$. Similarly, 
    \begin{align*}
        \Gamma_{\alpha,\lambda,\gamma_u}(m)=&1+\sum_{\substack{u\subseteq 1{:}s\\ u\neq\emptyset}}\gamma^{\frac{1}{\alpha+\lambda+1/2}}_u (A_{\alpha,\lambda}m +B_{\alpha,\lambda})^{|u|}\\
        \leq &\prod_{j=1}^s \Big(1+\gamma^{\frac{1}{\alpha+\lambda+1/2}}_j(A_{\alpha,\lambda}m +B_{\alpha,\lambda})\Big),
    \end{align*}
    which can be further bounded by $C'_{\eta} 2^{\eta m}$ for any $\eta>0$ with a constant $C'_{\eta}$ independent of $s$ \cite[Lemma 3]{HICKERNELL2003286}. Putting the above bounds into equation~\eqref{eqn:errorbound}, we conclude that equation~\eqref{eqn:erroboundasymp} holds with a constant $C_\eta$ independent of $s$. Notably, equation~\eqref{eqn:summable} resembles the tractability conditions assumed by \cite[Corollary 3.13]{Goda2024} for polynomial lattices to achieve near $O(2^{-(\alpha+\lambda+1/2)m})$ convergence rates.
\end{remark}

\section{Tractability under more general randomization}\label{sec:general}

While the completely random design offers theoretical advantages in near-optimal convergence rates and tractability, the random linear scrambling often outperforms it in practical settings due to tighter error bounds on low-dimensional ANOVA components. For example, \cite{joe:kuo:2008} designs the generating matrices $\mathcal{C}_j$ to explicitly optimize the discrepancy of all two-dimensional projections of the resulting digital nets. In this section, we extend the analysis of Section~\ref{sec:main} to a broader class of randomization schemes, with the random linear scrambling as a key example.

To state our assumption, we first define, for $d\in \natu$ and finite subset $\kappa\subseteq \natu$: 
$$\Vert\kappa\Vert_{(d)}=\begin{cases}
    \Vert \kappa\Vert, &\text{ if } |\kappa|\leq d \\
    \Vert \lceil\kappa\rceil_{1{:}d}\Vert, &\text{ if } |\kappa|>d
\end{cases}.$$

\begin{assumption}\label{assump:tu}
Under the randomization scheme for $C_j\in \{0,1\}^{E\times m}$, $j\in 1{:}s$, there exists $d\in \natu$ and a set of nonnegative integers $\tu$ for each nonempty $u\subseteq 1{:}s$ such that for $\bsk\in\natu^s_*$ with $\supp(\bsk)=u$, 
    $$\Pr(Z(\bsk)=1)\leq2^{-m+\tu} \bsone\Big\{\sum_{j\in u} \Vert\kappa_j\Vert_{(d)} > d(m-\tu)\Big\}.$$
\end{assumption}

From \cite[Lemma 2.17]{Goda2024}, the random linear scrambling satisfies the above assumption with parameters $d=1$ and $\Tilde{t}_u=t_u+|u|-1$, where $t_u$ denotes the projection-dependent $t$-value defined in \cite[Section 2]{Goda2024}. We frame this assumption to align with broader applications in higher-order digital nets \cite{dick:2011}, where the parameter $d$ (the order of digital nets) may exceed $1$. This formulation ensures compatibility with both classical $d=1$ constructions and advanced higher-order constructions.

We first consider the $d\geq \alpha+\lambda$ case. The analysis is easy thanks to the good properties of our randomization.

\begin{theorem}\label{thm:larged}
    Suppose $\Vert f\Vert_{\alpha,\lambda}<\infty$ for $\alpha\in \natu_0$ and $\lambda\in (0,1]$. Then for $m\in \natu, r\in \natu$ and $\theta\in (0,1)$, a randomization scheme satisfying Assumption~\ref{assump:tu} with $d\geq \alpha+\lambda$ ensures
\begin{equation*}
    \Pr\Big(|\hat{\mu}^{(r)}_{\infty}-\mu|^2>\Gamma_{\alpha,\lambda,\gamma_u,\tu,\theta}\Vert f\Vert^2_{\alpha,\lambda}2^{-(2\theta(\alpha+\lambda)+1)m}\Big)\leq  2^{-r},
\end{equation*}
where $\gamma_u$ is given by equation~\eqref{eqn:gammau},  $\tu$ is given in Assumption~\ref{assump:tu} and
$$\Gamma_{\alpha,\lambda,\gamma_u,\tu,\theta}=16\sum_{\substack{u\subseteq 1{:}s\\ u\neq\emptyset}} \gamma^2_uC_{\alpha,\lambda,\theta}^{|u|} 2^{\tu(2\theta(\alpha+\lambda)+1)}$$
with $C_{\alpha,\lambda,\theta}$ given by equation~\eqref{eqn:Cconstant}.
\end{theorem}

\begin{proof}
Let $K$ be defined as in equation~\eqref{eqn:setK} with $T_{u,m}=(\alpha+\lambda)(m-\tu)$. For any $\bsk'=(k'_1,\dots,k'_s)\in B(\bsk_u,v)$ with $\bsk_u=(k_j,j\in u)\in \mathcal{K}_{u,v}(T_{u,m})$, $\Vert \kappa'_j\Vert_{(d)}=\Vert \kappa_j\Vert$ for $j\in u\setminus v$ because $|\kappa'_j|=|\kappa_j|\leq d$ and 
$$\Vert \kappa'_j\Vert_{(d)}\leq\frac{d}{\alpha+\lambda} \Big((\lambda-1)\lceil\kappa_j\rceil_{\alpha+1}+\Vert\kappa_j\Vert\Big)$$
for $j\in v$ because $\lceil\kappa'_j\rceil_{1{:}(\alpha+1)}=\kappa_j$. Hence,
$$\sum_{j\in u}\Vert \kappa'_j\Vert_{(d)}\leq \frac{d}{\alpha+\lambda}\Big((\lambda-1)\sum_{j\in u} \lceil\kappa_j\rceil_{\alpha+1}+\sum_{j\in u}\Vert\kappa_j\Vert\Big)\leq \frac{d}{\alpha+\lambda}T_{u,m}= d(m-\tu). $$
It follows from Assumption~\ref{assump:tu} that $\Pr(Z(\bsk)=1)=0$ for all $\bsk\in K$ and we can apply Lemma~\ref{lem:medianMSE} with $\delta=1/16$ to get
\begin{equation*}
    \Pr\Big(|\hat{\mu}^{(r)}_{\infty}-\mu|^2>16\sum_{\bsk\in\natu^s_*\setminus K} \Pr(Z(\bsk)=1) |\hat{f}(\bsk)|^2\Big)\leq 2^{-r}.
\end{equation*}
Assumption~\ref{assump:tu} also implies $\Pr(Z(\bsk)=1)\leq 2^{-m+\tu}$ if $\supp(\bsk)=u$.  An argument similar to equation~\eqref{eqn:ftofu} gives
\begin{align*}
    \sum_{\bsk\in\natu^s_*\setminus K} \Pr(Z(\bsk)=1)|\hat{f}(\bsk)|^2
   \leq &\sum_{\substack{u\subseteq 1{:}s\\ u\neq\emptyset}}2^{-m+\tu}\sum_{\bsk\in\natu_*^{s}\setminus K_u(T_{u,m})}|\hat{f}_u(\bsk)|^2 \\
   \leq &   2^{-m}\sum_{\substack{u\subseteq 1{:}s\\ u\neq\emptyset}} 2^{\tu}\frac{\Vert f\Vert^2_{u,\alpha,\lambda} }{4^{\theta T_{u,m}}}   C_{\alpha,\lambda,\theta}^{|u|}
\end{align*}
Our conclusion follows after plugging in $\Vert f\Vert^2_{u,\alpha,\lambda} =\gamma^2_u\Vert f\Vert^2_{\alpha,\lambda}$ and $T_{u,m}=(\alpha+\lambda)(m-\tu)$.
\end{proof}

\begin{remark}
Analogous to Corollary~\ref{cor:eventA}, the bound in Theorem~\ref{thm:larged} can be refined by selecting thresholds $T_{u,m}\geq (\alpha+\lambda)(m-\tu)$ while ensuring
$$\sum_{\substack{u\subseteq 1{:}s\\ u\neq\emptyset}}2^{-m+\tu} \Big(|K_u(T_{u,m})| - |K_u\big((\alpha+\lambda)(m-\tu)\big)|\Big)\leq \delta.$$
However, this refinement offers limited practical improvement: Lemma~\ref{lem:KuTbound} suggests $|K_u(T)|$ grows exponentially in $T$, so the margin $T_{u,m}- (\alpha+\lambda)(m-\tu)$ is at best $O(\alpha+\lambda)$. Consequently, the error bound improves by only a factor of $2^{O(\alpha+\lambda)}$, leaving the asymptotic convergence rate unchanged.
\end{remark}

Next, we consider the $d<\alpha+\lambda$ case. We first establish the counterpart of Lemma~\ref{lem:KuTbound}.
\begin{lemma}\label{lem:KuTbound2}
For $d\in \natu, T'\geq 0$ and nonempty $u\subseteq 1{:}s$, let
$$K'_u(T')=\Big\{\bsk\in \natu^s_*\mid \supp(\bsk)=u,\sum_{j\in u}\Vert\kappa_j\Vert_{(d)}>T'\Big\}.$$
Then for $\alpha\in \natu_0,\lambda\in (0,1]$ satisfying $\alpha+\lambda>d$, $\theta'\in (0,1)$ and $T \in [T', (\alpha+\lambda)T'/d]$,
\begin{align*}
  &|K_u(T)\cap K'_u(T')|\\
  \leq &2^{(T-T')/\beta} \Big(A_{\alpha,\lambda}\max\big((T-T')/\beta,1\big) +B_{\alpha,\lambda}\Big)^{|u|} + 
    2^{(T-\theta'T')/\beta} D^{|u|}_{\alpha,\lambda,d,\theta'},
\end{align*}
    where $\beta=\alpha+\lambda-d$, $A_{\alpha,\lambda},B_{\alpha,\lambda}$ are given by equation~\eqref{eqn:ABconstant}, and 
    \begin{equation}\label{eqn:constantD}
        D_{\alpha,\lambda,d,\theta'}=\frac{1}{d!(2^{(1-\theta')/\beta}-1)^{d}}\sum_{\alpha'=1}^{\alpha-d}  \frac{1}{(\alpha')!(2^{1/\beta}-1)^{\alpha'}}+\sum_{\alpha'=1}^{d}\frac{1}{(\alpha')!(2^{(1-\theta')/\beta}-1)^{\alpha'}},
    \end{equation}
    with the sum from $1$ to $\alpha-d$ interpreted as $0$ when $\alpha=d$.
\end{lemma}
\begin{proof}
   See Appendix~\ref{app2}.
\end{proof}

\begin{corollary}\label{cor:eventA2}
    For $m\in \natu$ and $\gamma_u$ defined by equation~\eqref{eqn:gammau}, let
    \begin{equation}\label{eqn:Gamma2}
     \Gamma_{\alpha,\lambda,d,\theta',\gamma_u}(m)=1+\sum_{\substack{u\subseteq 1{:}s\\ u\neq\emptyset}}\gamma^{\frac{1}{\alpha+\lambda+1/2}}_u \Big ((A_{\alpha,\lambda}m+B_{\alpha,\lambda})^{|u|}+2^{\frac{(1-\theta')d}{\alpha+\lambda-d}m}D^{|u|}_{\alpha,\lambda,d,\theta'}\Big).   
    \end{equation}
    Then for any $\delta\in (0,1)$, 
    $$\Pr\Big(Z(\bsk)=1 \text{ for some } \bsk\in \bigcup_{\substack{u\subseteq 1{:}s\\ u\neq\emptyset}} K_u(T_{u,m})\Big)\leq \delta$$
    under a randomization scheme  satisfying Assumption~\ref{assump:tu} with $d< \alpha+\lambda$, where
    \begin{equation}\label{eqn:Tum2}
        T_{u,m}=(\alpha+\lambda)(m-\tu)-(\alpha+\lambda-d)\Big(\log_2\big(\Gamma_{\alpha,\lambda,d,\theta',\gamma_u}(m)\big)-\log_2(\delta)-\frac{\log_2(\gamma_u)}{\alpha+\lambda+1/2}\Big).
    \end{equation}
\end{corollary}
\begin{proof}
Let $T'_{u,m}=d(m-\tu)$. By Assumption~\ref{assump:tu}, for any $\bsk\in \natu^s_*$ with  $\supp(\bsk)=u$, $\Pr(Z(\bsk)=1)=0$ if $\bsk\notin K'_u(T'_{u,m})$ and   $\Pr(Z(\bsk)=1)\leq 2^{-m+\tu}$ if $\bsk\in K'_u(T'_{u,m})$.
    Because $\Gamma_{\alpha,\lambda,d,\theta',\gamma_u}(m)\geq 1, \delta<1$ and $\sup_{u\subseteq 1{:}s}\gamma_u=1$, we have $T_{u,m}\leq (\alpha+\lambda)(m-\tu)= (\alpha+\lambda)T'_{u,m}/d$. If further $T_{u,m}\geq T'_{u,m}$, we can apply Lemma~\ref{lem:KuTbound2} to get
    \begin{align*}
  &|K_u(T_{u,m})\cap K'_u(T'_{u,m})|\\
  \leq &\frac{ \delta \gamma^{\frac{1}{\alpha+\lambda+1/2}}_u}{\Gamma_{\alpha,\lambda,d,\theta',\gamma_u}(m)} 2^{m-\tu}\Big(( A_{\alpha,\lambda}m +B_{\alpha,\lambda})^{|u|} + 
   2^{\frac{(1-\theta')d}{\alpha+\lambda-d}m} D^{|u|}_{\alpha,\lambda,d,\theta'}\Big).
\end{align*}
    A union bound over all $\bsk\in K_u(T_{u,m})\cap K'_u(T'_{u,m})$ implies
    \begin{align*}
     &\Pr(Z(\bsk)=1 \text{ for some } \bsk \in K_u(T_{u,m}))\\
     \leq &\frac{ \delta \gamma^{\frac{1}{\alpha+\lambda+1/2}}_u}{\Gamma_{\alpha,\lambda,d,\theta',\gamma_u}(m)}\Big(( A_{\alpha,\lambda}m +B_{\alpha,\lambda})^{|u|} + 
   2^{\frac{(1-\theta')d}{\alpha+\lambda-d}m} D^{|u|}_{\alpha,\lambda,d,\theta'}\Big).   
    \end{align*}
    On the other hand, when $T_{u,m}<T'_{u,m}$, every $\bsk\in K_u(T_{u,m})$ satisfies
    $$\sum_{j\in u}\Vert\kappa_j\Vert_{(d)}\leq \sum_{j\in u}\Vert\kappa_j\Vert\leq T_{u,m}<T'_{u,m}.$$
    Thus, $K_u(T_{u,m})\cap K'_u(T'_{u,m})=\emptyset$ and $\Pr(Z(\bsk)=1 \text{ for some } \bsk \in K_u(T_{u,m}))=0$. The conclusion follows by taking an union bound over all nonempty $u\subseteq 1{:}s$.
\end{proof}

The following theorem is the counterpart of Theorem~\ref{thm:CRD} when $d<\alpha+\lambda$.

\begin{theorem}\label{thm:smalld}
     Suppose $\Vert f\Vert_{\alpha,\lambda}<\infty$ for $\alpha\in \natu_0$ and $\lambda\in (0,1]$. Then for $m\in \natu, r\in \natu,\theta\in (0,1)$ and $\theta'\in (0,1)$, a randomization scheme satisfying Assumption~\ref{assump:tu} with $d<\alpha+\lambda$ ensures
\begin{equation*}
    \Pr\Big(|\hat{\mu}^{(r)}_{\infty}-\mu|^2>\Gamma_{\alpha,\lambda,d,\gamma_u,\tu,\theta}\Vert f\Vert^2_{\alpha,\lambda}\frac{(\Gamma_{\alpha,\lambda,d,\theta',\gamma_u}(m))^{2\theta(\alpha+\lambda-d)}}{2^{(2\theta(\alpha+\lambda)+1)m}}\Big)\leq  2^{-r},
\end{equation*}
where $\gamma_u$ is given by equation~\eqref{eqn:gammau},  $\tu$ is given in Assumption~\ref{assump:tu} and
$$\Gamma_{\alpha,\lambda,d,\gamma_u,\tu,\theta}=16^{2\theta(\alpha+\lambda-d)+1}\sum_{\substack{u\subseteq 1{:}s\\ u\neq\emptyset}}\gamma^{\frac{1+2\theta d+2(1-\theta)(\alpha+\lambda)}{\alpha+\lambda+1/2}}_u  C_{\alpha,\lambda,\theta}^{|u|}2^{\tu(2\theta(\alpha+\lambda)+1)} $$
with $C_{\alpha,\lambda,\theta}$ given by equation~\eqref{eqn:Cconstant}.
\end{theorem}
\begin{proof}
    Let $\delta=1/16$ in Lemma~\ref{lem:medianMSE} and Corollary~\ref{cor:eventA2}. By Lemma~\ref{lem:fukbound} and Assumption~\ref{assump:tu},
    $$\sum_{\bsk\in\natu_*^{s}\setminus K_u(T_{u,m})} \Pr(Z(\bsk)=1) |\hat{f}_u(\bsk)|^2 \leq 2^{-m+\tu}\frac{\Vert f\Vert^2_{u,\alpha,\lambda} }{4^{\theta T_{u,m}}}   C_{\alpha,\lambda,\theta}^{|u|}.$$
    From Equation~\eqref{eqn:fuWalsh} and \eqref{eqn:gammau},
    \begin{align}\label{eqn:conditionalvarbound}
      \sum_{\bsk\in\natu^s_*\setminus K} \Pr(Z(\bsk)=1) |\hat{f}(\bsk)|^2 = &\sum_{\substack{u\subseteq 1{:}s\\ u\neq\emptyset}}\sum_{\bsk\in\natu_*^{s}\setminus K_u(T_{u,m})} \Pr(Z(\bsk)=1) |\hat{f}_u(\bsk)|^2 \\
      \leq &\Vert f\Vert^2_{\alpha,\lambda} \sum_{\substack{u\subseteq 1{:}s\\ u\neq\emptyset}} 2^{-m+\tu}\frac{\gamma^2_u}{4^{\theta T_{u,m}}}   C_{\alpha,\lambda,\theta}^{|u|}. \nonumber
    \end{align}
    Our conclusion follows from Lemma~\ref{lem:medianMSE} after plugging in $T_{u,m}$ from equation~\eqref{eqn:Tum2}.
\end{proof}

\begin{remark}
    By selecting $\theta$ and $\theta'$ sufficiently close to $1$, an argument analogous to Remark~\ref{rmk:asymp} shows $\hat{\mu}^{(r)}_{\infty}$ achieves an asymptotic convergence rate of $O(2^{-(\alpha+\lambda+1/2-\eta)m})$ for any $\eta>0$.
\end{remark}

\begin{remark}\label{rmk:tractability2}
    Consider again settings where $s\to\infty$ while $\Vert f\Vert_{\alpha,\lambda}\leq 1$. We assume equation~\eqref{eqn:summable} holds for weights $\{\gamma_j,j\in \natu\}$. Setting $\theta'=1-\eta/d$ for $\eta>0$, the analysis in Remark~\ref{rmk:tractability} implies $\Gamma_{\alpha,\lambda,d,\theta',\gamma_u}(m)\leq C_\eta 2^{\eta m/(\alpha+\lambda-d)}$ for a constant $C_\eta$ independent of $s$. Therefore, $\hat{\mu}^{(r)}_{\infty}$ attains dimension-independent error bounds if $\Gamma_{\alpha,\lambda,d,\gamma_u,\tu,\theta}$ remains bounded as $s\to\infty$, which depends on growth of $\tu$.

    Recall that $\tu=t_u+|u|-1$. It is shown in \cite{c812a4a5-75ac-3614-936d-5782e44941d0} that
    the base-2 Niederreiter sequence \cite{NIEDERREITER198851} satisfies
    \begin{equation}\label{eqn:Niederseq}
        t_u\leq \sum_{j\in u}\Big(\log_2(j)+\log_2\log_2(j+2)+2\Big),
    \end{equation}
    while the Sobol' sequence \cite{sobol67} satisfies 
    $$t_u\leq \sum_{j\in u} \Big(\log_2(j)+\log_2\log_2(j+1)+\log_2\log_2\log_2(j+3)+C\Big)$$
    for an absolute constant $C$. 
    
    For simplicity, we continue our discussion with the assumption $  \tu\leq \sum_{j\in u }\tj$ for nonnegative numbers $\{\tj, j\in \natu\}$. We can then bound
    \begin{align*}
      \frac{\Gamma_{\alpha,\lambda,d,\gamma_u,\tu,\theta}}{16^{2\theta(\alpha+\lambda-d)+1}}  \leq & \sum_{\substack{u\subseteq 1{:}s\\ u\neq\emptyset}}\prod_{j\in u}\gamma^{\frac{1+2 d}{\alpha+\lambda+1/2}}_j  C_{\alpha,\lambda,\theta}2^{\tj(2(\alpha+\lambda)+1)}\\ 
   \leq & \prod_{j=1}^s\Big(1+\gamma^{\frac{1+2 d}{\alpha+\lambda+1/2}}_jC_{\alpha,\lambda,\theta}2^{\tj(2(\alpha+\lambda)+1)}\Big)\\
   \leq &  \exp\Big(C_{\alpha,\lambda,\theta}\sum_{j=1}^s \gamma^{\frac{1+2 d}{\alpha+\lambda+1/2}}_j2^{\tj(2(\alpha+\lambda)+1)}\Big).
    \end{align*}
    Therefore, $\Gamma_{\alpha,\lambda,d,\gamma_u,\tu,\theta}$ has a dimension-independent bound if 
    $$\sum_{j=1}^\infty \gamma^{\frac{1+2 d}{\alpha+\lambda+1/2}}_j2^{\tj(2(\alpha+\lambda)+1)}<\infty.$$
\end{remark}

\begin{remark}\label{rmk:tractability3}
   Viewing $\Gamma_{\alpha,\lambda,\gamma_u,\tu,\theta}$ from Theorem~\ref{thm:larged} as a special case of $\Gamma_{\alpha,\lambda,d,\gamma_u,\tu,\theta}$ with $d=\alpha+\lambda$, the argument in Remark~\ref{rmk:tractability2} yields dimension-independent error bounds for the $d\geq \alpha+\lambda$ case if $\gamma_u\leq \prod_{j\in u} \gamma_j$, $  \tu\leq \sum_{j\in u }\tj$ and
    $$\sum_{j=1}^\infty \gamma_j^{2}2^{\tj(2(\alpha+\lambda)+1)}<\infty.$$
While increasing $d$ may appear to enhance tractability, parameters $\tu$ and $\tj$ typically scale with $d$, as in the interlacing construction \cite{dick:2011}, and unbounded increases in $d$ ultimately violate the tractability condition above. Thus, prudent selection of $d$ is critical for practical implementations.

\end{remark}

When $d<\alpha+\lambda$, we can refine our argument to establish dimension-independent convergence under a weaker condition analogous to that outlined in Remark~\ref{rmk:tractability3}.

\begin{theorem}\label{thm:tractability}
    Suppose $\Vert f\Vert_{\alpha,\lambda}<\infty$ for $\alpha\in \natu_0$ and $\lambda\in (0,1]$. Given  a randomization scheme satisfying Assumption~\ref{assump:tu} with $d< \alpha+\lambda$, we further assume $\gamma_u\leq \prod_{j\in u} \gamma_j$ for $\{\gamma_j,j\in \natu\}$, $  \tu\leq \sum_{j\in u }\tj$ for $\{\tj,j\in \natu\}$ and 
    \begin{equation}\label{eqn:gammautu}
      \sum_{j=1}^\infty\Big(\gamma^{\frac{1}{\alpha+\lambda+1/2}}_j 2^{\tj}\Big)^{\phi}<\infty \quad \text{for} \quad \phi=\frac{\alpha+\lambda+1/2}{\alpha+\lambda-d+1/2}.
    \end{equation}
    Then under the given randomization scheme, we can find a constant $C_\eta$ independent of $s$ for any $\eta>0$ such that for all $m\in\natu$ and $r\in \natu$,
    $$\Pr\Big(|\hat{\mu}^{(r)}_{\infty}-\mu|^2 > C_\eta \Vert f\Vert^2_{\alpha,\lambda} 2^{-(2\alpha+2\lambda+1-2\eta)m}\Big)\leq 2^{-r}.$$
\end{theorem} 
\begin{proof}
By equation~\eqref{eqn:gammautu}, only finitely many $j\in \natu$ satisfy $\gamma^{\frac{1}{\alpha+\lambda+1/2}}_j 2^{\tj}\geq 1$. Hence,
   $$\sup_{\substack{u\subseteq 1{:}s\\u\neq\emptyset}} \gamma^{\frac{1}{\alpha+\lambda+1/2}}_u 2^{\tu}\leq \sup_{\substack{u\subseteq 1{:}s\\u\neq\emptyset}} \prod_{j\in u } \gamma^{\frac{1}{\alpha+\lambda+1/2}}_j 2^{\tj} \leq \Lambda$$
   for a constant $\Lambda$ independent of $s$. Let $\delta'=\min(\Lambda^{-\phi},1/16)$, $\theta'=1-\eta /d,$
    $$\Gamma(m)=1+\sum_{\substack{u\subseteq 1{:}s\\ u\neq\emptyset}}\Big(\gamma^{\frac{1}{\alpha+\lambda+1/2}}_u  2^{\tu}\Big)^{\phi}\Big ((A_{\alpha,\lambda}m+B_{\alpha,\lambda})^{|u|}+2^{\frac{(1-\theta')d}{\alpha+\lambda-d}m}D^{|u|}_{\alpha,\lambda,d,\theta'}\Big)$$
    and
    \begin{align*}
      T_{u,m}=&(\alpha+\lambda)(m-\tu)-\\
      &(\alpha+\lambda-d)\Big(\log_2(\Gamma(m))-\log_2(\delta')-\frac{\phi\log_2(\gamma_u)}{\alpha+\lambda+1/2}-\phi \tu\Big).  
    \end{align*}
    By our choice of $\delta'$, 
    $$\log_2(\delta')+\frac{\phi\log_2(\gamma_u)}{\alpha+\lambda+1/2}+\phi \tu =\log_2(\delta')+\phi\log_2\Big(\gamma^{\frac{1}{\alpha+\lambda+1/2}}_u 2^{\tu} \Big) \leq 0.$$
    Therefore, $T_{u,m}\leq (\alpha+\lambda)(m-\tu)$ and an argument similar to that used in Corollary~\ref{cor:eventA2} shows
    \begin{align*}
        &\Pr\Big(Z(\bsk)=1 \text{ for some } \bsk\in \bigcup_{\substack{u\subseteq 1{:}s\\ u\neq\emptyset}} K_u(T_{u,m})\Big)\\
        \leq & \frac{\delta'}{\Gamma(m)} \sum_{\substack{u\subseteq 1{:}s\\ u\neq\emptyset}} \Big(\gamma^{\frac{1}{\alpha+\lambda+1/2}}_u 2^{\tu}\Big)^{\phi}\Big(( A_{\alpha,\lambda} m+B_{\alpha,\lambda})^{|u|} + 
  2^{\frac{(1-\theta')d}{\alpha+\lambda-d}m} D^{|u|}_{\alpha,\lambda,d,\theta'}\Big),
    \end{align*}
   which is bounded by $\delta'$ due to the definition of $\Gamma(m)$. Equation~\eqref{eqn:conditionalvarbound} with the above choice of $T_{u,m}$ then implies 
   \begin{align*}
       \sum_{\bsk\in\natu^s_*\setminus K} \Pr(Z(\bsk)=1) |\hat{f}(\bsk)|^2
       \leq &\Vert f\Vert^2_{\alpha,\lambda} \sum_{\substack{u\subseteq 1{:}s\\ u\neq\emptyset}} 2^{-m+\tu}\frac{\gamma^2_u}{4^{\theta T_{u,m}}}   C_{\alpha,\lambda,\theta}^{|u|}\\
       = &\Gamma'\Vert f\Vert^2_{\alpha,\lambda}\frac{(\Gamma(m)/\delta')^{2\theta(\alpha+\lambda-d)}}{2^{(2\theta(\alpha+\lambda)+1)m}},
   \end{align*}
   where
   $$\Gamma'=\sum_{\substack{u\subseteq 1{:}s\\ u\neq\emptyset}} \gamma^{2-\frac{2\theta\phi(\alpha+\lambda-d)}{\alpha+\lambda+1/2}}_uC_{\alpha,\lambda,\theta}^{|u|} 2^{\tu\big(1+2\theta(\alpha+\lambda)-2\theta\phi(\alpha+\lambda-d)\big)}.$$
   Analogous to Remark~\ref{rmk:tractability2}, we can show $\Gamma(m)\leq C'_\eta 2^{\eta m/(\alpha+\lambda-d)}$ for a constant $C'_\eta$ independent of $s$ when equation~\eqref{eqn:gammautu} holds. Furthermore, we notice that $\alpha+\lambda>\phi (\alpha+\lambda-d)$ and
   \begin{align*}
       \Gamma' \leq   &\sum_{\substack{u\subseteq 1{:}s\\ u\neq\emptyset}} \gamma^{2-\frac{2\phi(\alpha+\lambda-d)}{\alpha+\lambda+1/2}}_uC_{\alpha,\lambda,\theta}^{|u|} 2^{\tu\big(1+2(\alpha+\lambda)-2\phi(\alpha+\lambda-d)\big)}\\
       =&\sum_{\substack{u\subseteq 1{:}s\\ u\neq\emptyset}}\Big(\gamma^{\frac{1}{\alpha+\lambda+1/2}}_u2^{\tu}\Big)^\phi  C^{|u|}_{\alpha,\lambda,\theta},
   \end{align*}
   which is bounded as $s\to\infty$ because of equation~\eqref{eqn:gammautu}. Our conclusion then follows from Lemma~\ref{lem:medianMSE} with $\delta=1/16$.
\end{proof}

\begin{remark}
By Hölder's inequality,
    \begin{align*}
        &\sum_{j=1}^\infty\Big(\gamma^{\frac{1}{\alpha+\lambda+1/2}}_j 2^{\tj}\Big)^{\phi}\\
\leq &\Big(\sum_{j=1}^\infty \gamma^{\frac{1}{\alpha+\lambda+1/2}}_j\Big)^{1-\frac{\phi}{1+2(\alpha+\lambda)}}\Big(\sum_{j=1}^\infty \gamma^{\frac{1+2d}{\alpha+\lambda+1/2}}_j2^{\tj(2(\alpha+\lambda)+1)}\Big)^{\frac{\phi}{1+2(\alpha+\lambda)}}.
    \end{align*}
    Therefore, equation~\eqref{eqn:gammautu} holds under the assumptions in Remark~\ref{rmk:tractability2}.
\end{remark}

\begin{remark}\label{rmk:tractability4}
For base-2 Niederreiter sequences, substituting the bound \eqref{eqn:Niederseq} into equation \eqref{eqn:gammautu} yields the simplified tractability condition:
    $$\sum_{j=1}^\infty\Big(\gamma^{\frac{1}{\alpha+\lambda+1/2}}_j j\log_2(j)\Big)^{\phi}<\infty.$$
    Since $\phi>1$, this condition is strictly weaker than the requirements in \cite[Corollary 3.11]{Goda2024} for such digital nets to achieve near $O(2^{-(\alpha+\lambda+1/2)m})$ convergence rates under the random linear scrambling.

    In general, when $\gamma_j=O(j^{-p})$ and $2^{\tj}=O(j^q)$ for $p\geq 0$ and $q\geq 0$, $\hat{\mu}^{(r)}_\infty$ attains dimension-independent convergence if
    \begin{equation}\label{eqn:tractability}
     \phi(\frac{p}{\alpha+\lambda+1/2}-q)=\frac{p-(\alpha+\lambda+1/2)q}{\alpha+\lambda-d+1/2}>1.   
    \end{equation}
    This criterion interpolates between the tractability condition for the completely random design (by setting $d=0$ and $q=0$) and that for randomization schemes with $d\geq \alpha+\lambda$ (by setting $d=\alpha+\lambda$).
\end{remark}

\section{Numerical experiments}\label{sec:exper}

This section tests our theoretical results on integrands with varying dimensions and decay rates of $\gamma_u$. We also reference \cite{pan2024automatic}, which provides further experiments demonstrating near-optimal convergence rates of median RQMC and comparisons to higher-order scrambled digital nets.

The test integrands have the common form
$$f_{s,\gamma,\alpha}(\bsx)=\prod_{j=1}^s \Big(1+\frac{c_\alpha}{j^\gamma} (x_j\exp(x_j)-1)\Big),$$
where $s=10,100,1000$ controls the dimension of integration domain, $\gamma=2,3,4$ controls the decay rate of $\gamma_u$ and $c_\alpha$ for $\alpha=0,1$ is a normalizing constant to ensure $\Vert f\Vert_{\alpha,1}=1$. To determine $c_\alpha$, we start by computing the ANOVA component given by equation~\eqref{eqn:fudef}:
\begin{equation}\label{eqn:fuexample}
 f_u(\bsx)=\frac{c^{|u|}_\alpha}{(\prod_{j\in u}j)^{\gamma}}\prod_{j\in u}(x_j\exp(x_j)-1).   
\end{equation}
For $\lambda=1$ and smooth function $f$, the Vitali variation defined in equation~\eqref{eqn:variationdef} can be viewed as a Riemann sum for 
$$\Big(\int_{[0,1]^d} \Big(f^{(1,\dots,1)}(\bsx)\Big)^2\rd \bsx\Big)^{1/2}.$$
Hence, $V^{(|u|)}_{1}(f_u)=\Vert f^{(1,\dots,1)}_u\Vert_{L^2}$. As discussed in Section~\ref{sec:anova}, $ \Vert f\Vert_{u,\alpha,\lambda}=V^{u}_\lambda (f_u)$ when $\alpha=0$. Therefore,
$$ \Vert f\Vert_{u,0,1}=\Vert f^{(1,\dots,1)}_u\Vert_{L^2}=\frac{c^{|u|}_0}{(\prod_{j\in u}j)^{\gamma}}\Big(\int_0^1 (x+1)^2\exp(2x)\rd x\Big)^{|u|/2}$$
and we can make $\Vert f_{s,\gamma,0}\Vert_{0,1}=1$ by setting 
$$c_0=\Big(\int_0^1 (x+1)^2\exp(2x)\rd x\Big)^{-1/2}.$$
Next, when $\alpha=1$, it follows from the definition~\eqref{eqn:funorm} that
$$\Vert f\Vert_{u,1,1}=\sup_{v\subseteq u}\sup_{\substack{\bsalpha_u\in \ints_{\le 1}^{|u|}\\ \alpha_j=1 \ \forall j\in v,\\ \alpha_j>0 \ \forall j\in u\setminus v  }} V^{v}_1 (f^{(\bsalpha_u)}_u)=\sup_{v\subseteq u}V^{v}_1 (f^{(1,\dots,1)}_u).$$
Let 
$g=f^{(1,\dots,1)}_u.$
For each $v\subseteq u$, because $((x+1)\exp(x))'=(x+2)\exp(x)$,
$$V^{(|v|)}_1\Bigl( g(\text{ }\cdot\text{ };\bsx_{v^c})\Bigr)=\frac{c^{|u|}_1}{(\prod_{j\in u}j)^{\gamma}}\Big(\int_0^1 (x+2)^2\exp(2x)\rd x\Big)^{|v|/2} \prod_{j\in u\setminus v} (x_j+1)\exp(x_j).$$
Therefore, by equation~\eqref{eqn:Vvbound}, 
\begin{align*}
 &\Big(V^{v}_1(g)\Big)^2\\
 \leq &\frac{c^{2|u|}_1}{(\prod_{j\in u}j)^{2\gamma}}\Big(\int_0^1 (x+2)^2\exp(2x)\rd x\Big)^{|v|} \Big(\int_0^1 (x+1)^2\exp(2x)\rd x\Big)^{|u|-|v|}   \\
 \leq & \Vert f^{(2,\dots,2)}_u\Vert_{L^2}^2,
\end{align*}
while $V^{u}_1(g)=\Vert g^{(1,\dots,1)}\Vert_{L^2}= \Vert f^{(2,\dots,2)}_u\Vert_{L^2}$. Therefore,
$$\Vert f\Vert_{u,1,1}=\Vert f^{(2,\dots,2)}_u\Vert_{L^2}=\frac{c^{|u|}_1}{(\prod_{j\in u}j)^{\gamma}}\Big(\int_0^1 (x+2)^2\exp(2x)\rd x\Big)^{|u|/2}$$
and we can make $\Vert f_{s,\gamma,1}\Vert_{1,1}=1$ by setting 
$$c_1=\Big(\int_0^1 (x+2)^2\exp(2x)\rd x\Big)^{-1/2}.$$
With the above choice of $c_\alpha$ for $\alpha=0,1$, we also have $\gamma_u=\prod_{j\in }\gamma_j$ for $\gamma_j=j^{-\gamma}$.

For each integrand $f_{s,\gamma,\alpha}$, we evaluate the performance of the median estimator $\hat{\mu}^{(r)}_E$ for $m=1,2,\dots,16$ under two randomization schemes: the completely random design (CRD) and the random linear scrambling (RLS). The estimator $\hat{\mu}^{(r)}_E$ is computed as the sample median of $2r-1$ independent realizations of $\hat{\mu}_E$, where $\hat{\mu}_E$ is defined as in equation~\eqref{eqn:muEdef} with precision $E=64$. The observed integration errors in our tested range of $m$ is much larger than $2^{-E}$, ensuring that $\hat{\mu}^{(r)}_E$ closely approximates $\hat{\mu}^{(r)}_\infty$. Consequently, the error bounds in Theorem~\ref{thm:CRD}, \ref{thm:larged} and \ref{thm:smalld} apply to $\hat{\mu}^{(r)}_E$ with probability at least $1-2^{-r}$. We set $r=10$, limiting the failure probability per data point to below $0.1\%$. As a benchmark, we compute the standard RQMC estimator (STD), defined as the average of the $2r-1$ $\hat{\mu}_E$ realizations used by RLS. The generating matrices $\mathcal{C}_j$ used by RLS and STD come from \cite{joe:kuo:2008}.

\begin{figure}[htb]
  \setkeys{Gin}{width=\linewidth}
  \setlength\tabcolsep{2pt}
  \begin{tabularx}{\textwidth}{XXX}
    \includegraphics[width=0.3\textwidth]{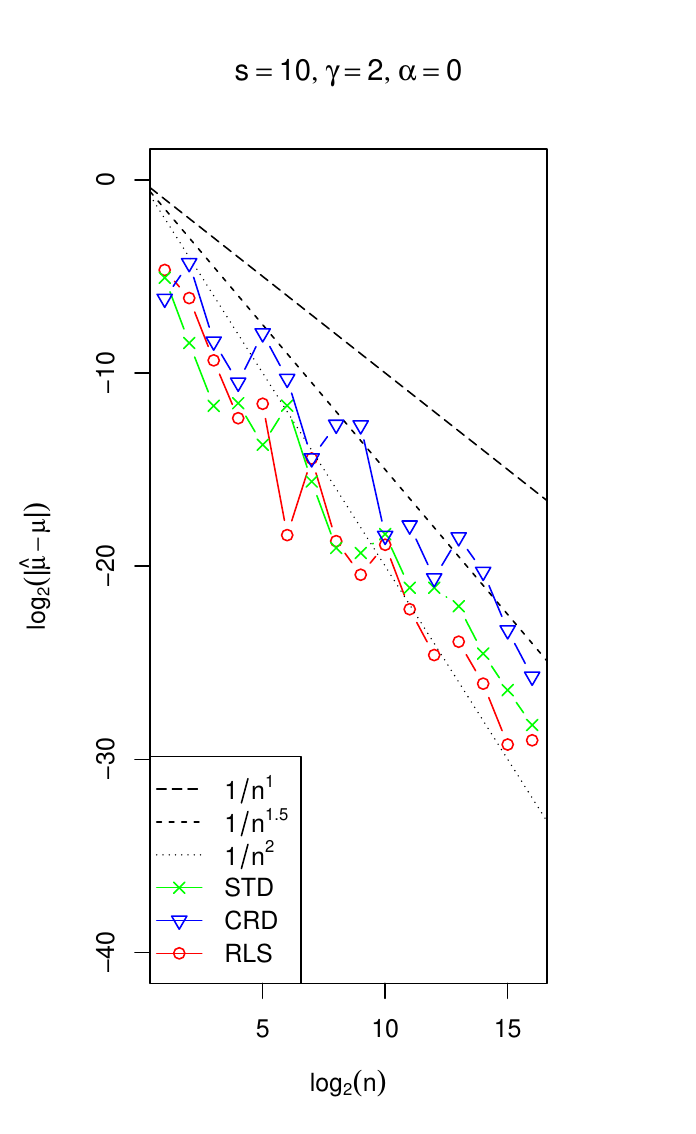} &
    \includegraphics[width=0.3\textwidth]{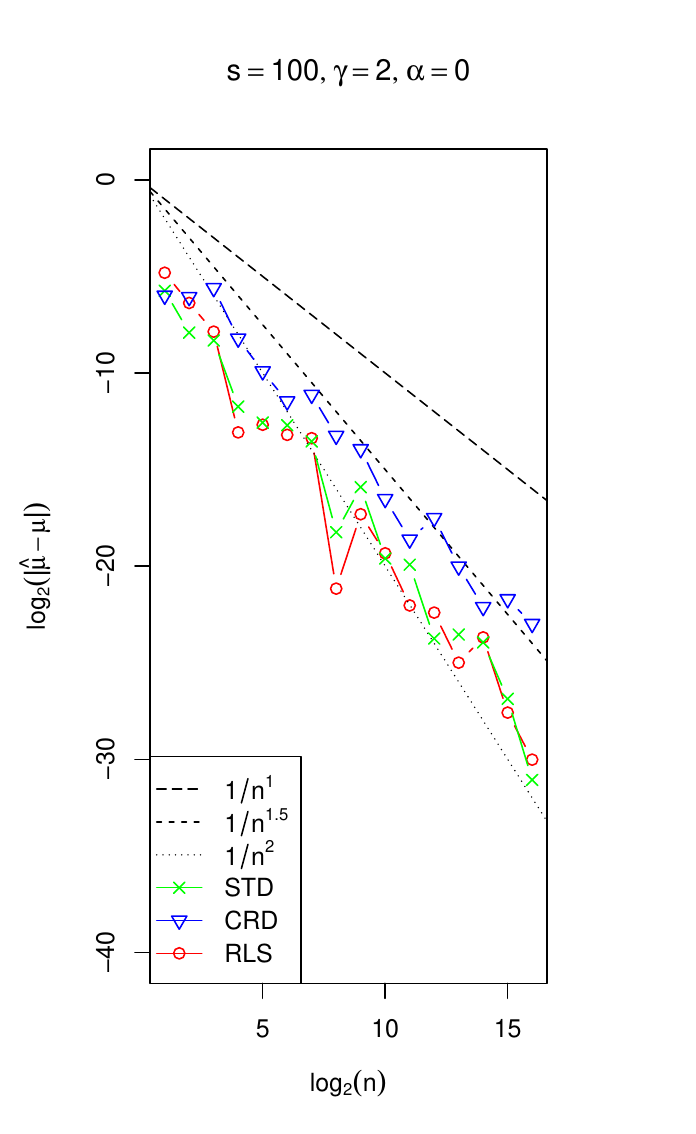} &
    \includegraphics[width=0.3\textwidth]{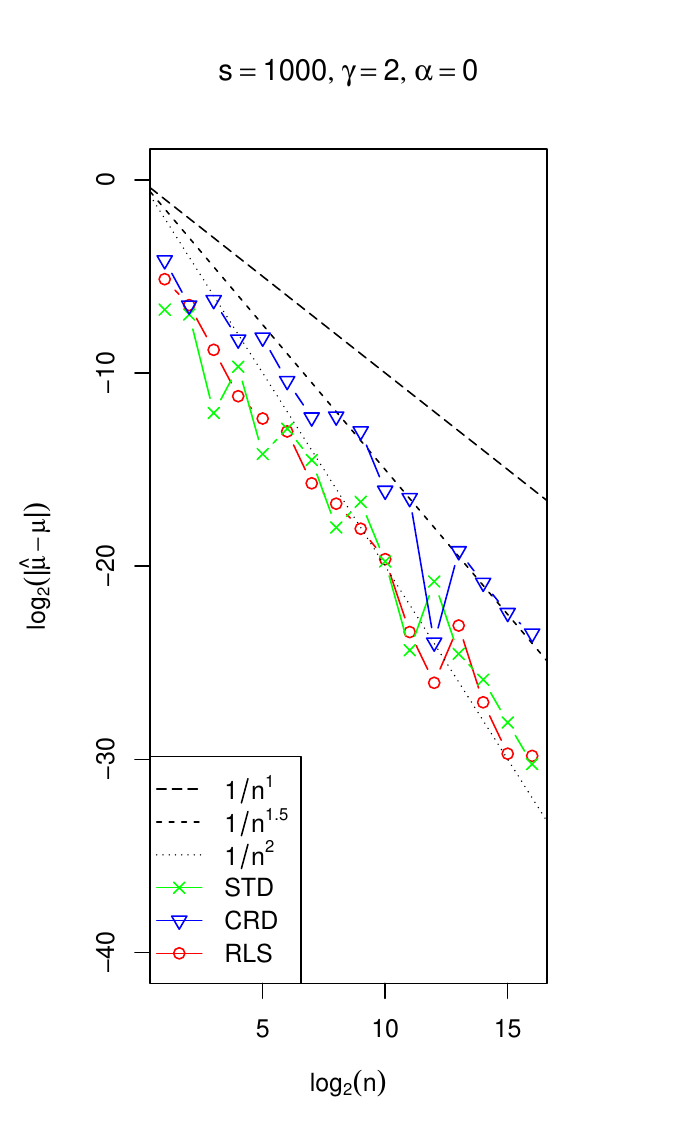}   \\
    \includegraphics[width=0.3\textwidth]{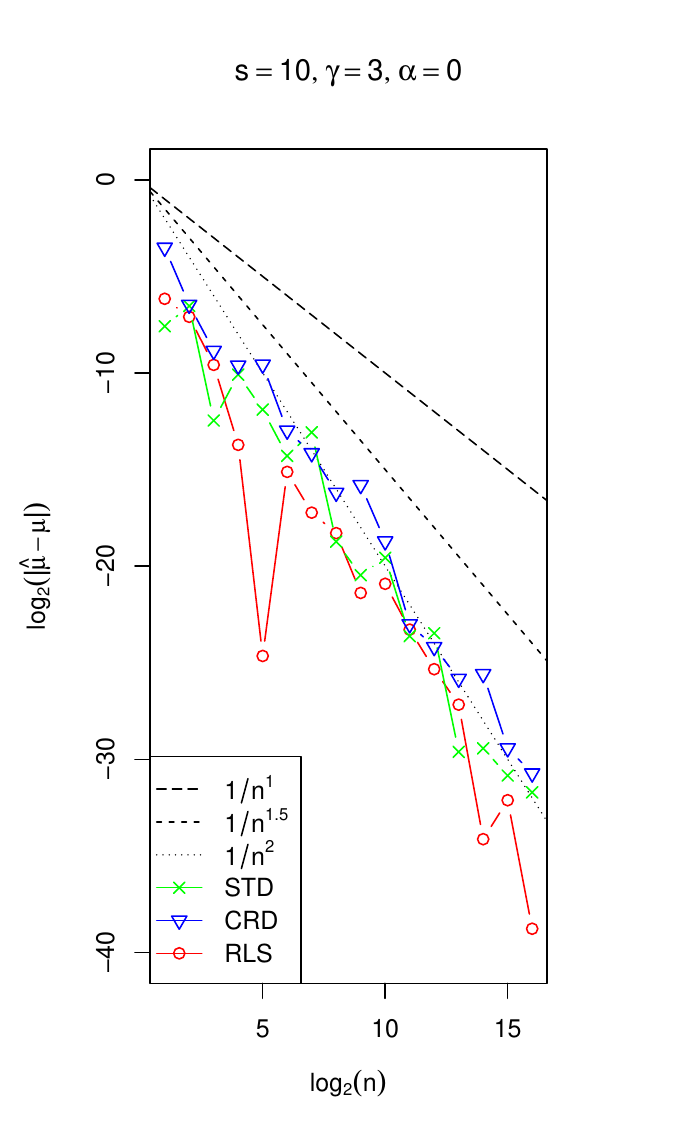} &
    \includegraphics[width=0.3\textwidth]{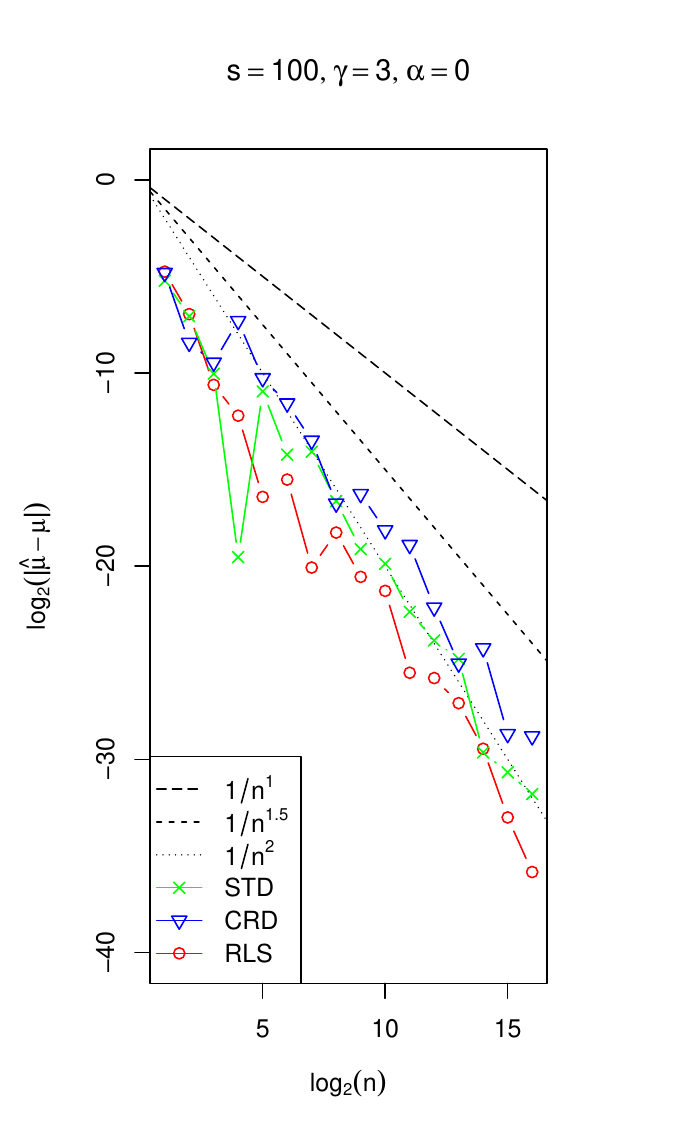} &
    \includegraphics[width=0.3\textwidth]{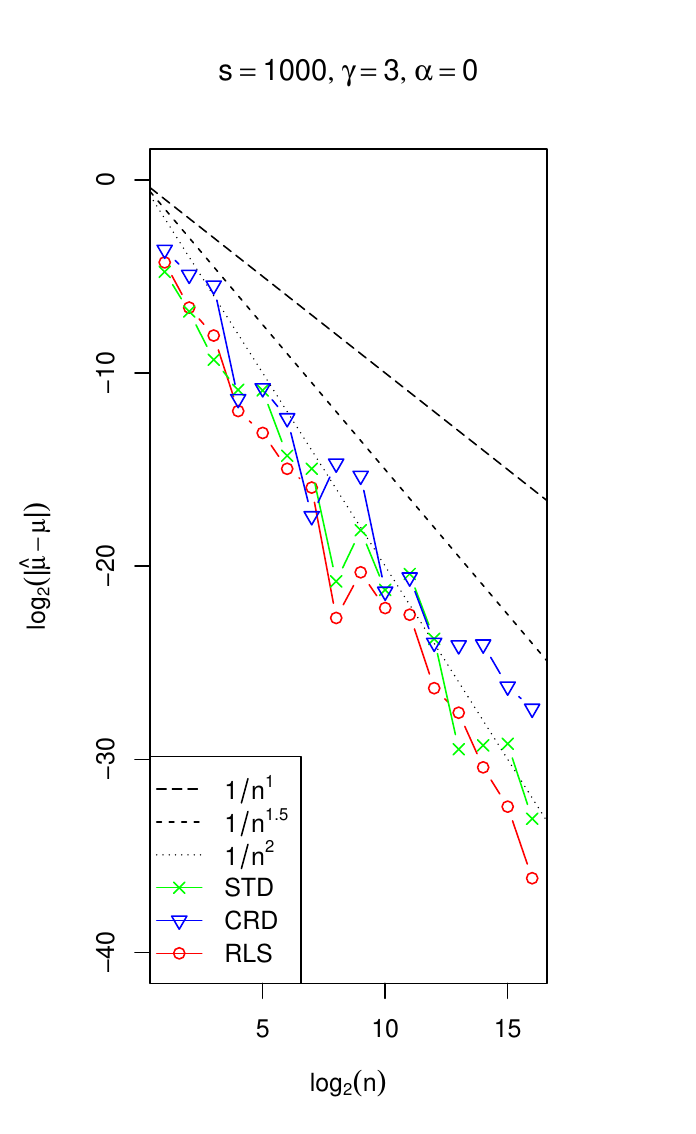}   \\
    \includegraphics[width=0.3\textwidth]{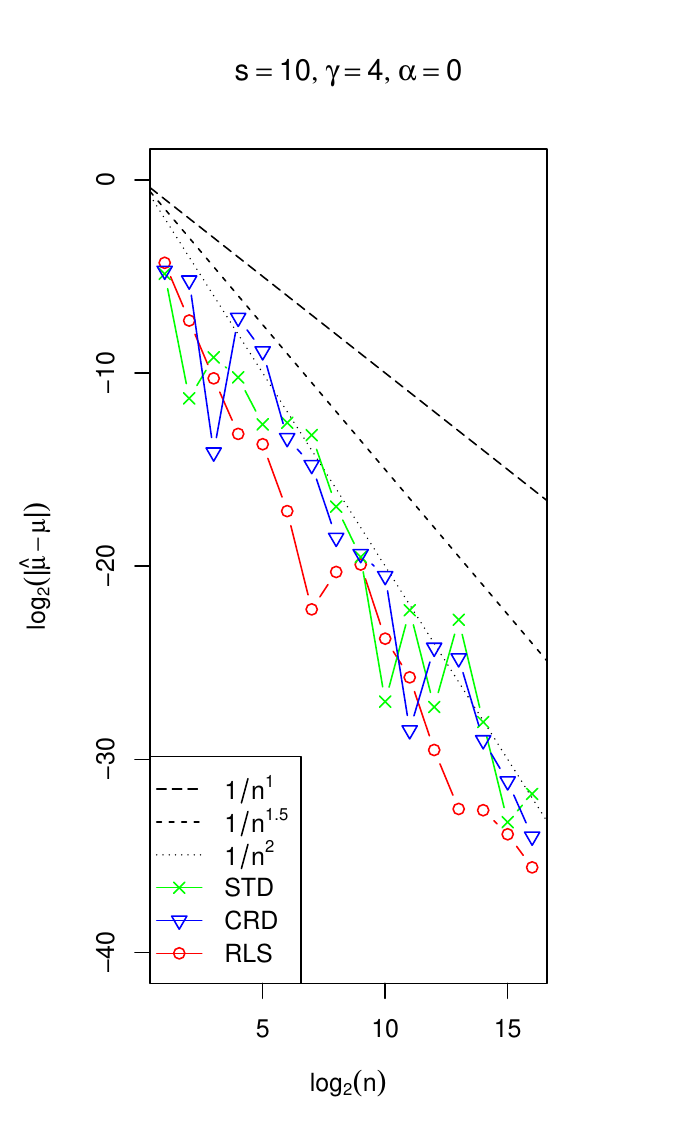} &
    \includegraphics[width=0.3\textwidth]{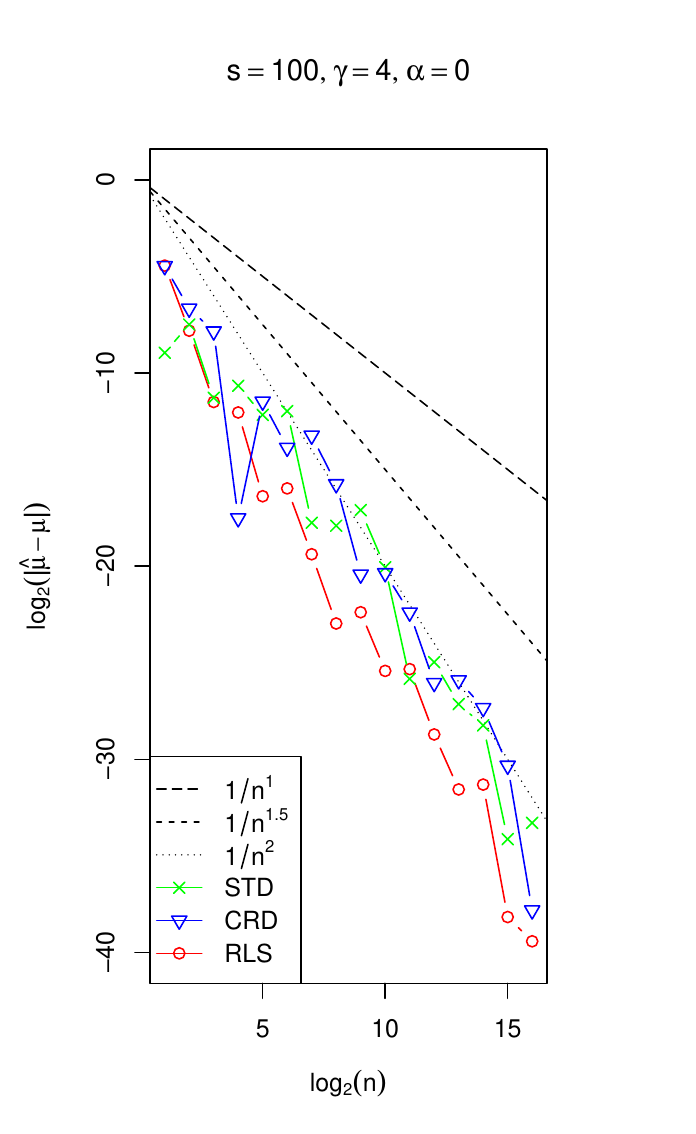} &
    \includegraphics[width=0.3\textwidth]{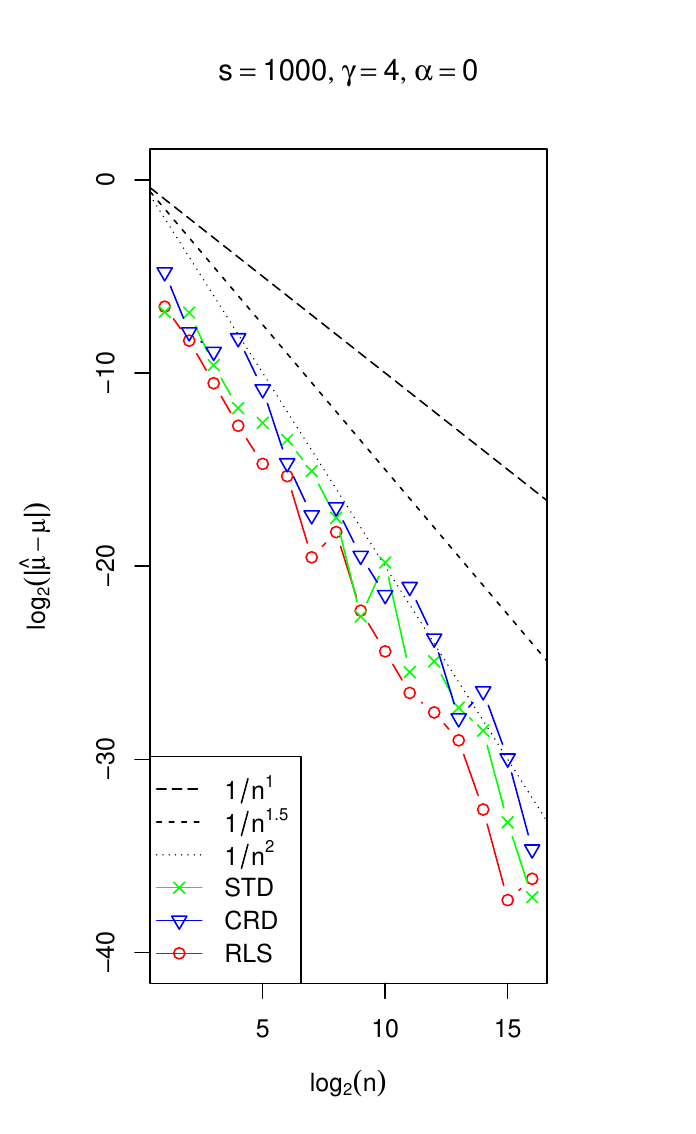}
  \end{tabularx}
\caption{Integration errors for $f_{s,\gamma,\alpha}$ with $\alpha=0$.}
\label{fig:alpha0}
\end{figure}

\begin{figure}[htb]
  \setkeys{Gin}{width=\linewidth}
  \setlength\tabcolsep{2pt}
  \begin{tabularx}{\textwidth}{XXX}
    \includegraphics[width=0.3\textwidth]{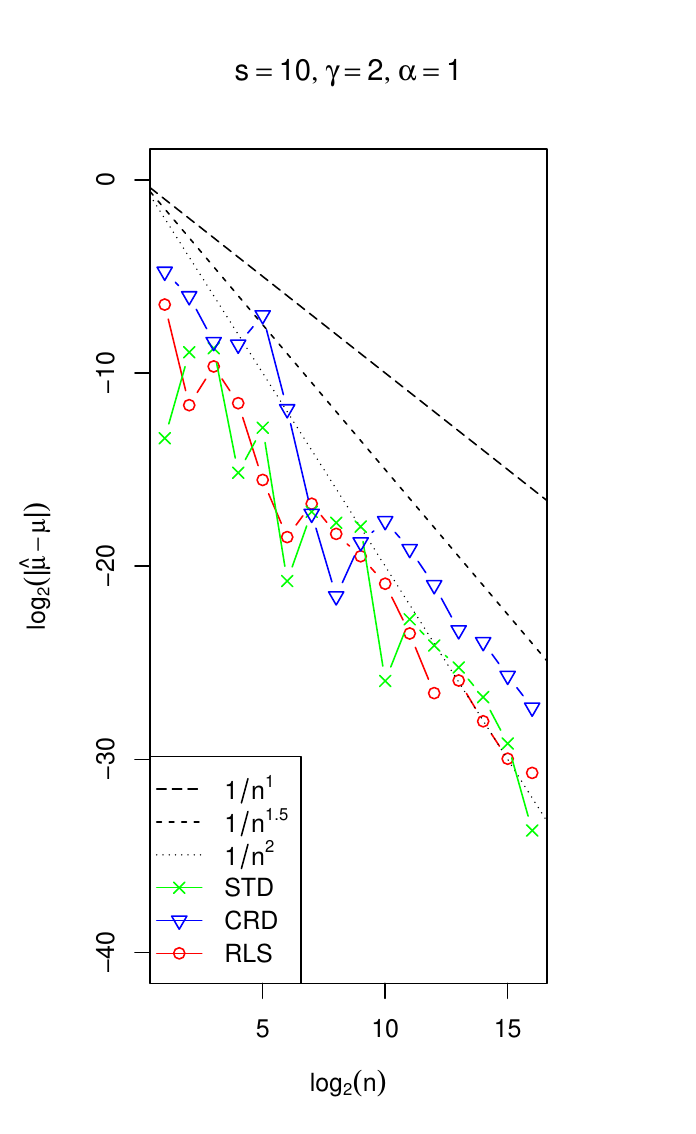} &
    \includegraphics[width=0.3\textwidth]{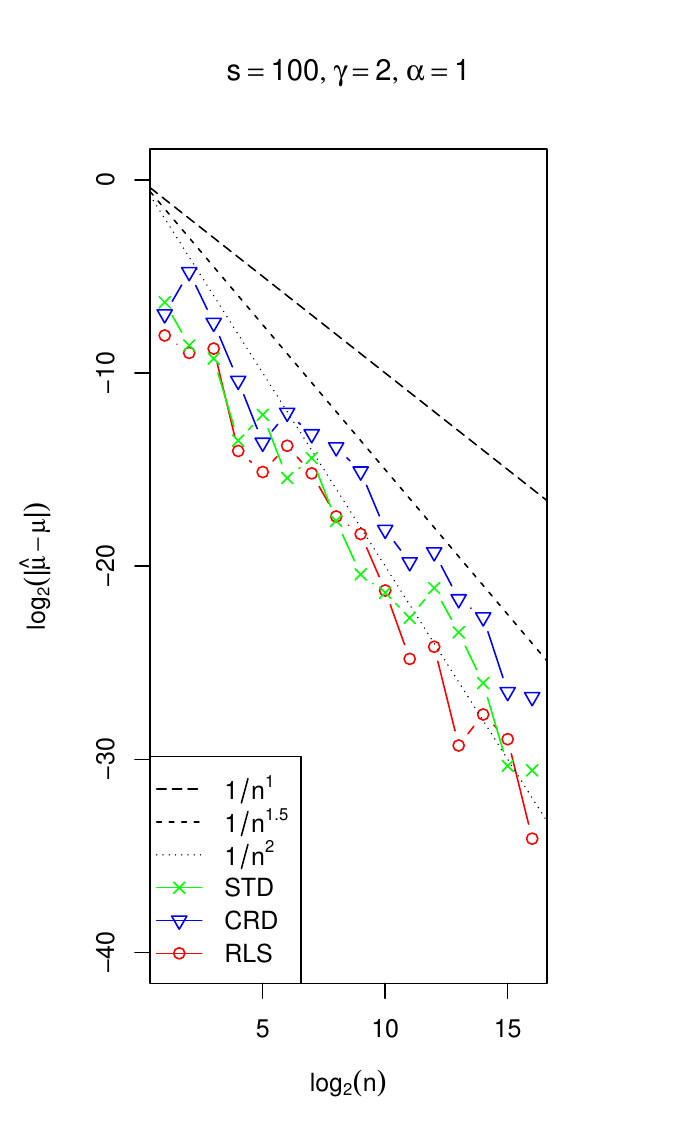} &
    \includegraphics[width=0.3\textwidth]{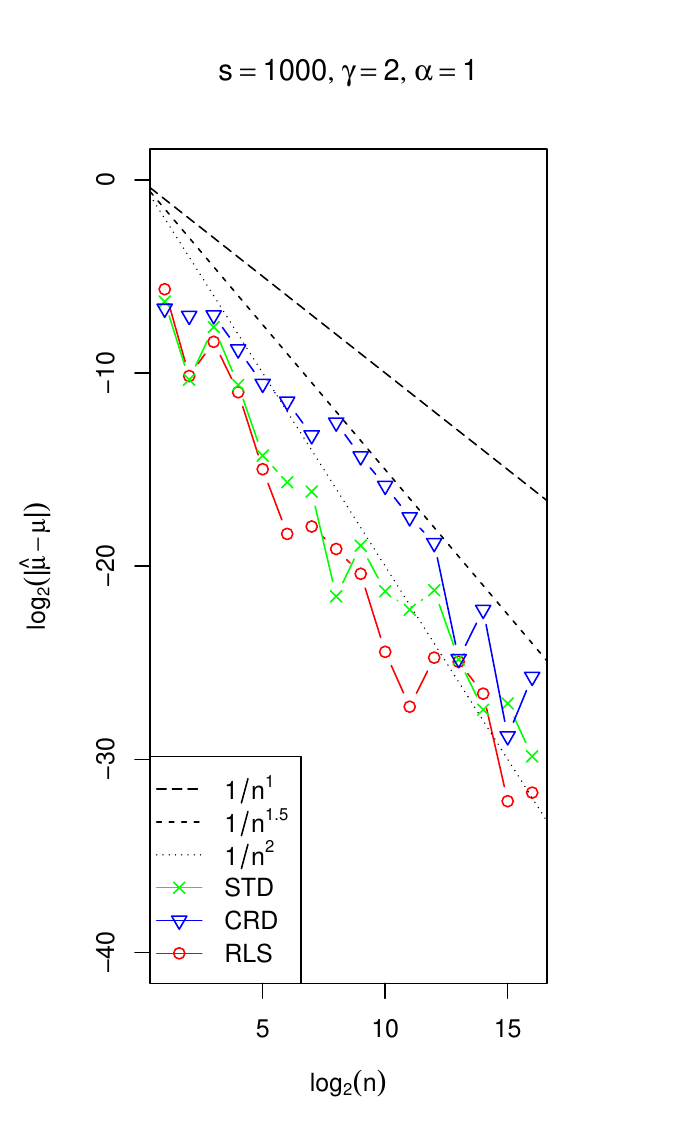}   \\
    \includegraphics[width=0.3\textwidth]{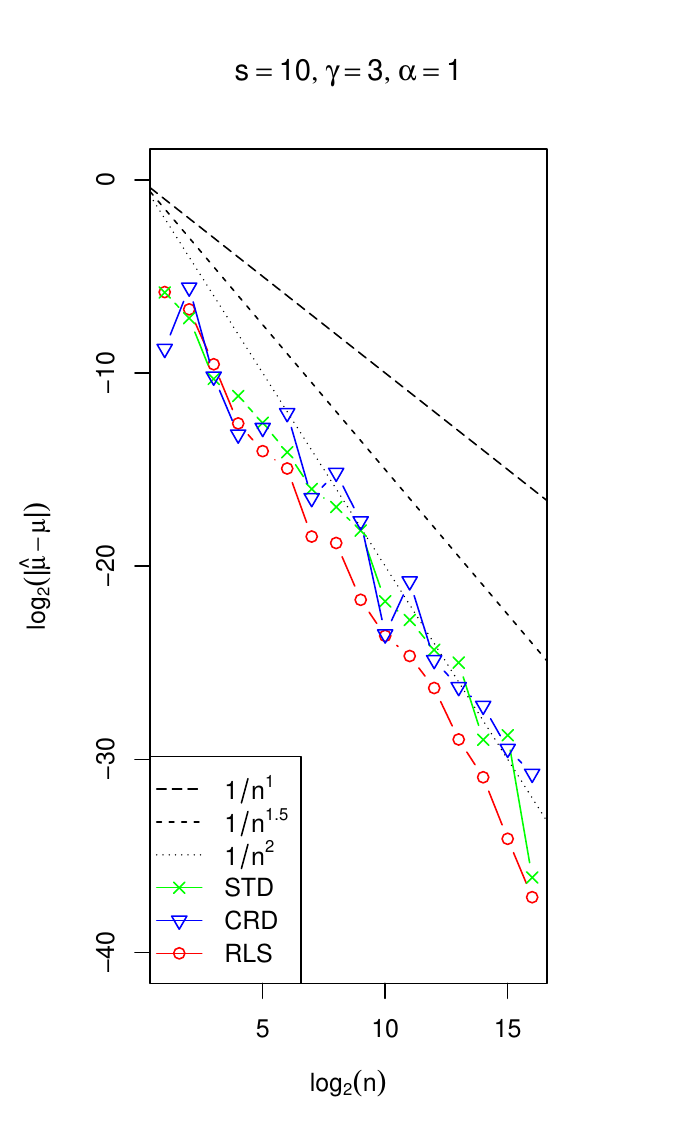} &
    \includegraphics[width=0.3\textwidth]{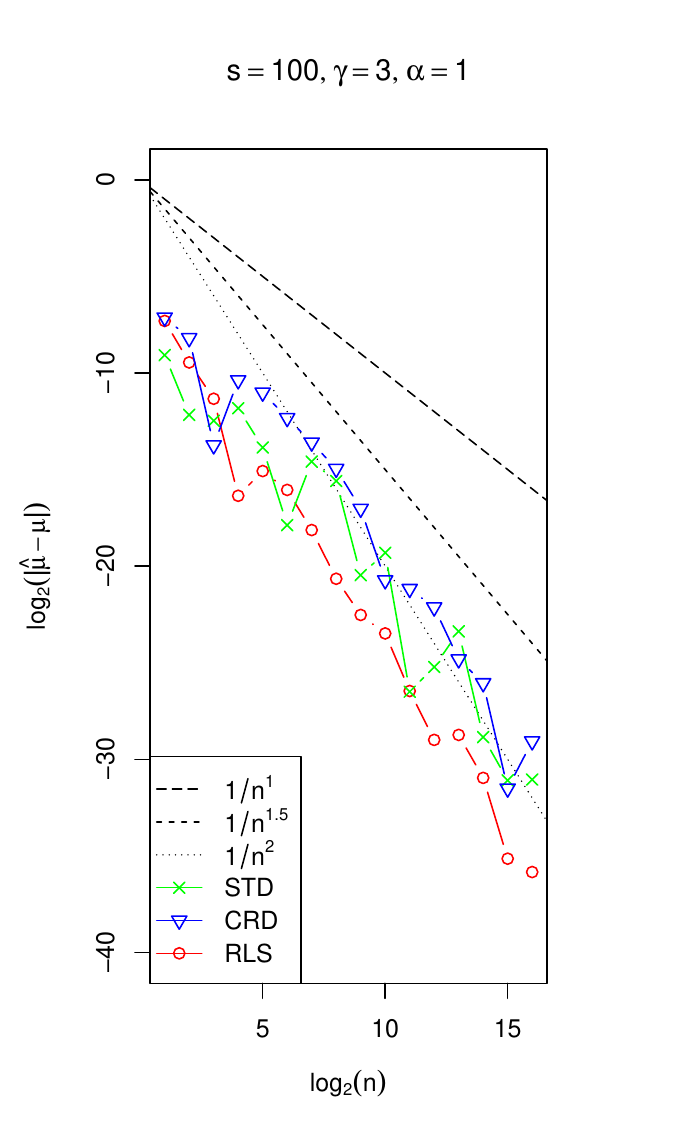} &
    \includegraphics[width=0.3\textwidth]{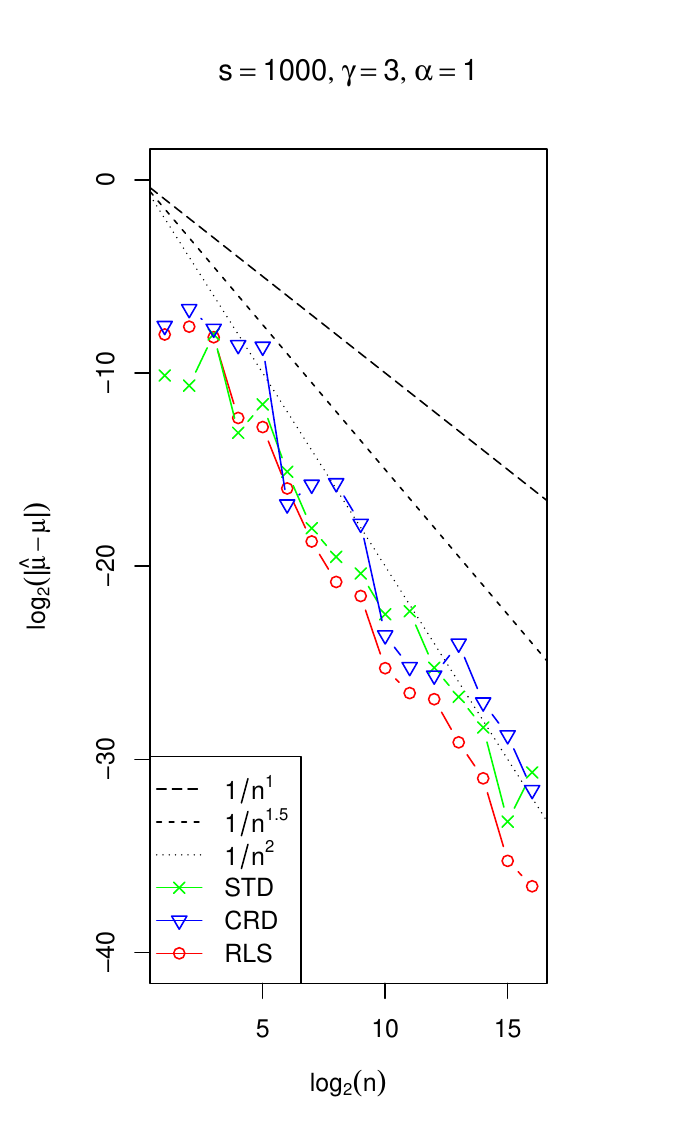}   \\
    \includegraphics[width=0.3\textwidth]{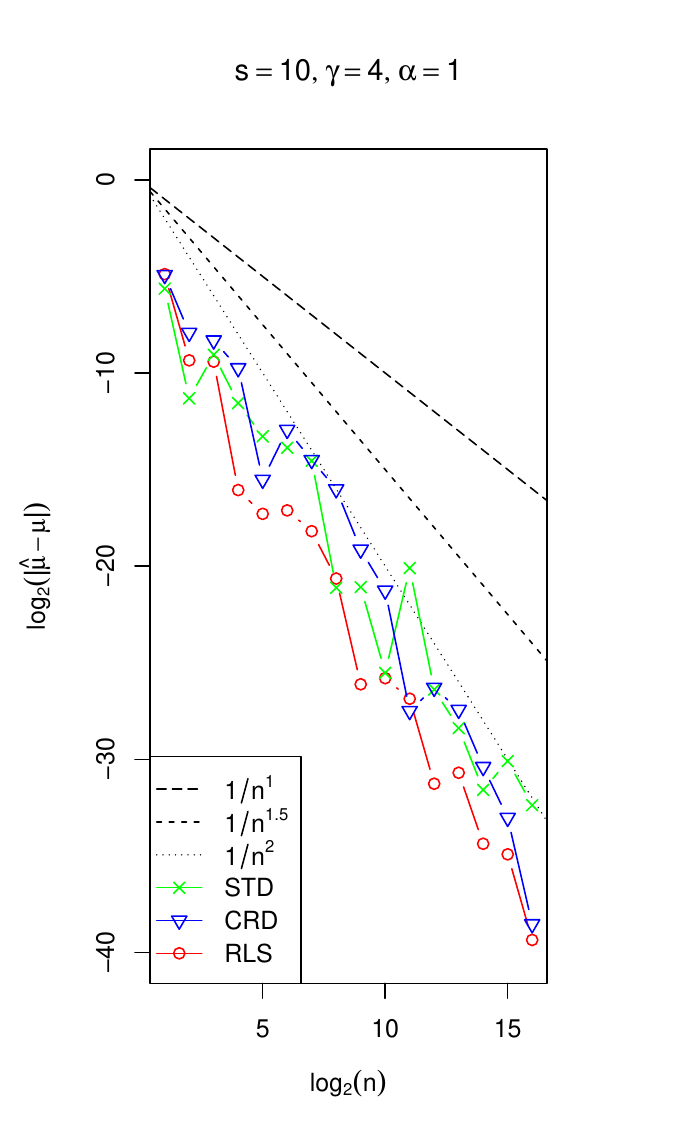} &
    \includegraphics[width=0.3\textwidth]{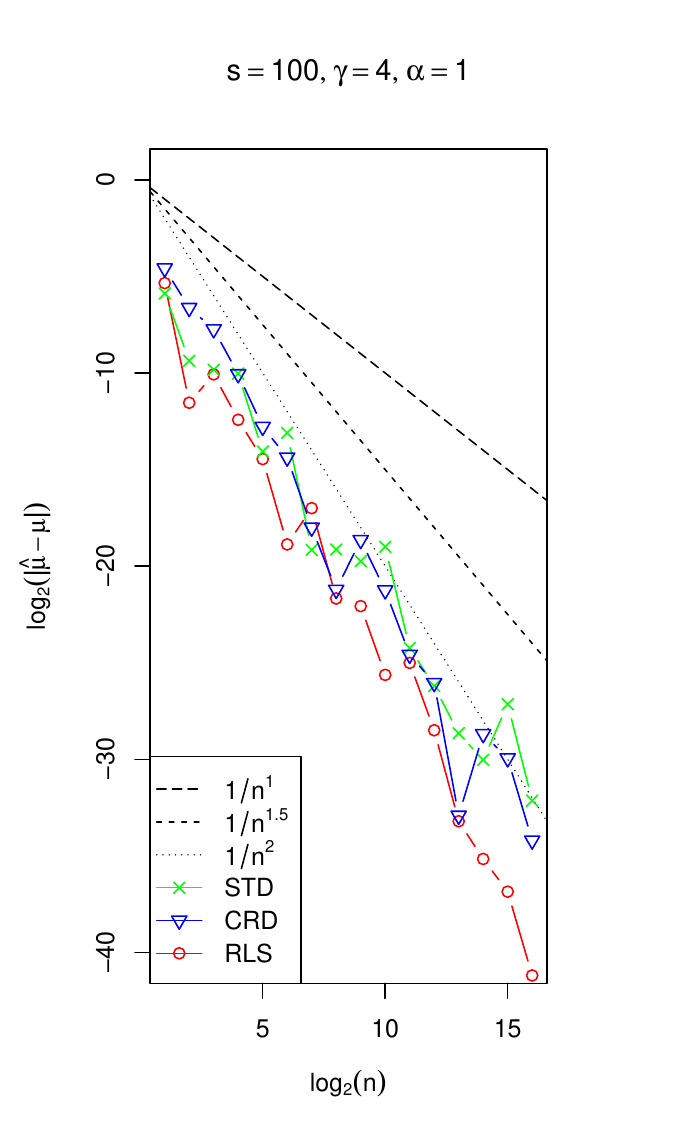} &
    \includegraphics[width=0.3\textwidth]{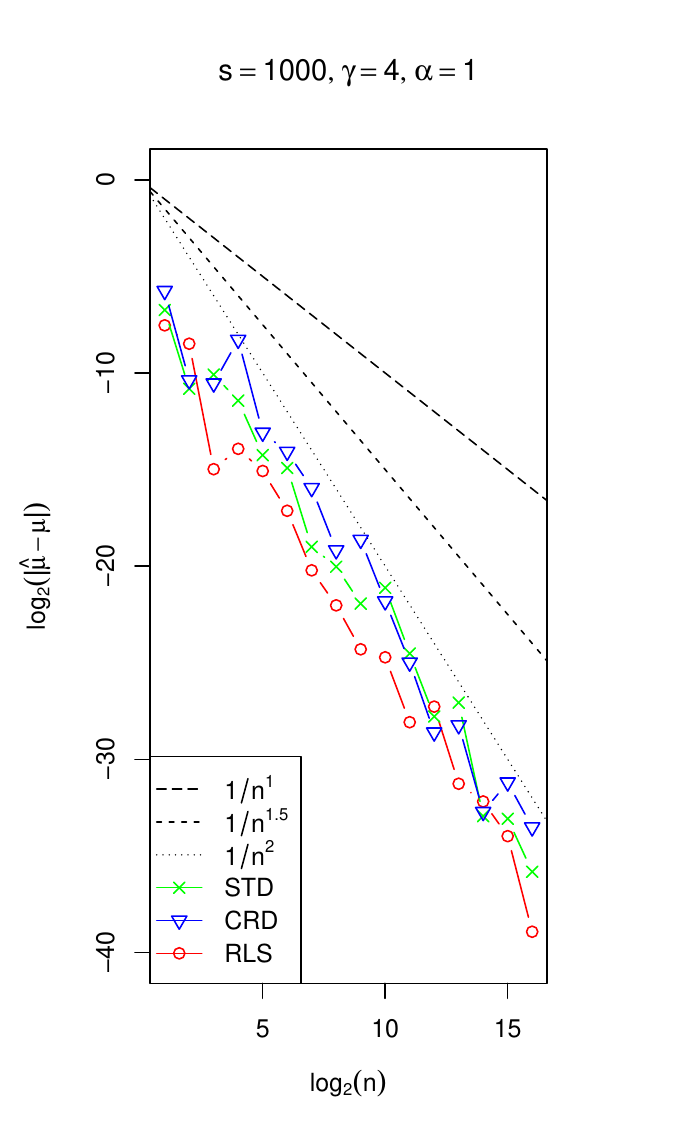}
  \end{tabularx}
\caption{Integration errors for $f_{s,\gamma,\alpha}$ with $\alpha=1$.}
\label{fig:alpha1}
\end{figure}

Figure~\ref{fig:alpha0} and Figure~\ref{fig:alpha1} summarize the simulation results. First, we observe that errors for $\alpha=0$ and $\alpha=1$ exhibit nearly identical convergence rates. Surprisingly, for $\alpha=0$ and $\gamma=3$ or $4$, the empirical rate approaches $O(2^{-2m})$, despite the theoretical bound $O(2^{-(\alpha+\lambda+1/2)m})=O(2^{-(3/2)m})$. This phenomenon arises because $f_{s,\gamma,0}$ also has a finite $\Vert \cdot\Vert_{1,1}$ norm. By equation~\eqref{eqn:fuexample}, the ANOVA component $f_u$ of $f_{s,\gamma,0}$ equals $(c_0/c_1)^{|u|}$ times that of $f_{s,\gamma,1}$. Consequently,
$$\Vert f_{s,\gamma,0}\Vert_{u,1,1}=\Big(\frac{c_0}{c_1}\Big)^{|u|}\Vert f_{s,\gamma,1}\Vert_{u,1,1}=\Big(\frac{c_0}{c_1}\Big)^{|u|}\prod_{j\in u}j^{-\gamma}$$
and 
$$\Vert f_{s,\gamma,0}\Vert_{1,1}=\sup_{u\subseteq 1{:}s} \Vert f_{s,\gamma,0}\Vert_{u,1,1}=\sup_{u\subseteq 1{:}s} \prod_{j\in u}\Big(\frac{c_0}{c_1}j^{-\gamma}\Big)<\infty$$
as $j^{-\gamma}> c_1/c_0$ holds for only finitely many $j\in \natu$. Furthermore, because $\Vert f_{s,\gamma,0}\Vert_{1,1}\geq 1$, the relative variation $\gamma'_u$ of $f_{s,\gamma,0}$ under the $\Vert \cdot\Vert_{1,1}$ norm satisfies
$$\gamma'_u=\frac{\Vert f_{s,\gamma,0}\Vert_{u,1,1}}{\Vert f_{s,\gamma,0}\Vert_{1,1}}\leq \prod_{j\in u}\Big(\frac{c_0}{c_1}j^{-\gamma}\Big)=\prod_{j\in u}\gamma'_j$$
for $\gamma'_j=(c_0/c_1)j^{-\gamma}$. In general, this analysis extends to arbitrary norms $\Vert \cdot\Vert_{\alpha',\lambda'}$ with $\alpha'\in \natu_0,\lambda'\in (0,1]$, and the corresponding $\gamma_j$ is always proportional to $j^{-\gamma}$. Thus, the decay rate $\gamma$ dominates convergence behavior as  $s\to\infty$, irrespective of $\alpha$. 

Next, we observe that for a fixed value of $\gamma$, the integration errors exhibit similar behavior as $s$ varies, aligning with our analysis of dimension-independent convergence. Applying equation~\eqref{eqn:tractability} with parameters $d=1$, $p=\gamma$ and $q=1+\eta$ for an arbitrarily small $\eta>0$,  our theory predicts that for large $s$, $\hat{\mu}^{(r)}_\infty$ under RLS achieve near $O(2^{-(\alpha+\lambda+1/2)m})$ convergence rates, provided $\alpha+\lambda+1/2<(\gamma+1)/2$. Empirically, the observed rates approximate $O(2^{-(3/2)m})$ when $\gamma=2$ and $O(2^{-2m})$ when $\gamma=3$. For $\gamma=4$, the convergence rates slightly exceed $O(2^{-2m})$ but remain significantly slower than $O(2^{-(5/2)m})$. This discrepancy is likely due to the limited range of $m$ tested. Notably, while our theory suggests that $\hat{\mu}^{(r)}_\infty$ under CRD can attain near $O(2^{-(\alpha+\lambda+1/2)m})$ rates for $\alpha+\lambda+1/2<\gamma$, practical results show that CRD consistently underperforms RLS. We hypothesize that the superior theoretical rates for CRD may only manifest at substantially larger $m$.

Finally, we compare the integration errors under RLS and STD. Across the vast majority of test cases, RLS demonstrates superior performance over STD, reinforcing that the median RQMC approach (via RLS) is a more robust alternative to traditional RQMC (via STD). We also note that STD exhibits convergence rates comparable to those of RLS and substantially surpasses the theoretical $O(2^{-(3/2)m})$  rate derived from the root mean squared error when $\gamma=3$ or $4$. This discrepancy arises because each estimator $\hat{\mu}_\infty$ in STD is a special case of $\hat{\mu}^{(r)}_\infty$ with $r=1$, which satisfies the rapidly converging probabilistic error bounds in Theorems~\ref{thm:larged} and \ref{thm:smalld}. However, averaging multiple replicates in STD compromises these bounds, as even a single outlier among the replicates can dominate the error and inflate its magnitude. This is the key mechanism that motivates  our shift toward median RQMC approaches.

\section{Discussion}\label{sec:disc}

Our current analysis presumes fixed constants $\alpha$ and $\lambda$ for which $\Vert f\Vert_{\alpha,\lambda}<\infty$. In practice, integrands are expected to satisfy $\Vert f\Vert_{\alpha,\lambda}<\infty$ across a spectrum of $\alpha$ and $\lambda$ values, as demonstrated by $f_{s,\gamma,\alpha}$ in Section~\ref{sec:exper}.  Crucially, since median RQMC methods do not explicitly incorporate smoothness information, the resulting error bound achieves the minimum over all viable $\alpha$ and $\lambda$. A more nuanced analysis should consequently incorporate variable smoothness parameters and derive error bounds under the optimal ones. Developing a framework for such an analysis is an open problem we leave for future research.

Another unresolved challenge is how to select $d$ for our randomization. As noted in Remark~\ref{rmk:tractability3}, $\tu$ generally scales with $d$. Theorem~\ref{thm:larged} further implies that increasing $d$ beyond $ \alpha+\lambda$ does not enhance the convergence rate, suggesting the upper bound $d\leq\lceil \alpha+\lambda\rceil$. This constraint, however, offers little insights for high-dimensional integration of smooth integrands, where dimensionality—rather than smoothness—governs the convergence rate. A critical open question is how to rigorously define an effective smoothness criterion that jointly incorporates dimensionality and traditional smoothness parameters. Addressing this could yield a principled framework for selecting $d$, thereby optimizing integration performance.

A natural extension of our analysis is to integrands with discontinuities in their derivatives or function values. As discussed in Subsection~\ref{subsec:smoothness}, the requirement $f\in C^{(\alpha,\dots,\alpha)}([0,1]^s)$  can be relaxed provided that the weak derivatives satisfy equation~\eqref{eqn:exactfk}. Notably, the fractional Vitali variation framework inherently accommodates discontinuities. By developing a theory for non-smooth integrands, our results could be extended to real-world applications in finance and related fields. For instance, \cite{griebel2017anova} showcases an infinite-dimensional integration problem where all ANOVA components are smooth despite discontinuities in the integrand's derivatives, suggesting that non-smoothness may not fundamentally compromise convergence rates in high-dimensional settings.

\appendix

\section{Proof of Lemma~\ref{lem:fkbound}}\label{app3}
Let $\bsk'\in B(\bsk_u,v)$. Because $|\kappa'_j|\geq |\kappa_j|=\alpha+1$ for $j\in v$ and $|\kappa'_j|= |\kappa_j|=\alpha_j$ for $j\in u\setminus v$, we can apply Lemma~\ref{lem:exactfk} to get
\begin{align*}
  &|\hat f_u(\bsk')|^2=\\&\Big|\int_{[0,1]^{|u|}}\Big(f_u^{(\bsalpha_u)}(\bsx)\prod_{j\in u\setminus v} W_{\kappa_j}(x_j)\Big)\Big(\prod_{j\in v}  \walkappa{\kappa'_j\setminus \lceil\kappa_j\rceil_{1{:}\alpha}}(x_j)W_{\lceil\kappa_j\rceil_{1{:}\alpha}}(x_j)\Big)\rd \bsx_u\Big|^2,  
\end{align*}
Let $\kappa^+_j=\lceil\kappa_j\rceil_{1{:}\alpha}$,  $\ell_j=\lceil\kappa_j\rceil_{\alpha+1}$ and $\kappa^-_j=\kappa'_j\setminus \kappa_j$ for $j\in v$. We can rewrite the above equation as
\begin{equation*}
    |\hat f_u(\bsk')|^2=\Bigl|\int_{[0,1]^{|v|}}g(\bsx_v)\prod_{j\in v} \walkappa{\kappa^-_j}(x_j)(-1)^{\vec{x}_j(\ell_j)}W_{\kappa^+_j}(x_j)\rd \bsx_v\Bigr|^2,
\end{equation*}
where 
    $$g(\bsx_v)=I_v\Big(f_u^{(\bsalpha_u)}\prod_{j\in u\setminus v}W_{\kappa_j}\Big).$$
    In the special case $v=1{:}s$, \cite[page 14 line 14]{pan2024automatic} shows
    \begin{multline*}
       \sum_{\kappa^-_1,\dots,\kappa^-_s}\Bigl|\int_{[0,1]^{s}}g(\bsx)\prod_{j\in 1{:}s} \walkappa{\kappa^-_j}(x_j)(-1)^{\vec{x}_j(\ell_j)}W_{\kappa^+_j}(x_j)\rd \bsx\Bigr|^2 \\
\leq 2^{-s}\Bigl(V^{(s)}_\lambda(g)\Bigr)^2\Bigl(\prod_{j=1}^s 4^{-\lambda \ell_j}\prod_{\ell'\in \kappa^+_j}4^{-\ell'-1}\Bigr), 
    \end{multline*}
where the summation $\sum_{\kappa^-_1,\dots,\kappa^-_s}$ means each $\kappa^-_j$ runs over all subsets of $\natu$ satisfying $\lceil\kappa^-_j\rceil_1<\ell_j$. After mapping $v$ to $1{:}|v|$ and compactly rewriting the summation as $\sum_{\bsk'\in B(\bsk_u,v)}$, we obtain
\begin{equation}\label{eqn:fkg}
    \sum_{\bsk'\in B(\bsk_u,v)}|\hat{f}_u(\bsk')|^2\leq 2^{-|v|} \Big(V^{(|v|)}_\lambda(g)\Big)^2\Bigl(\prod_{j\in v} 4^{-\lambda \ell_j}\prod_{\ell'\in \kappa^+_j}4^{-\ell'-1}\Bigr).
\end{equation}
Letting $\rho_j(x_j)=W_{\kappa_j}(x_j)/\Vert W_{\kappa_j}\Vert_\infty$ for $j\in u\setminus v$, equation~\eqref{eqn:Wint} implies 
\begin{align*}
    V^{(|v|)}_\lambda(g)=&\Big(\prod_{j\in u\setminus v}2\prod_{\ell'\in \kappa_j}2^{-\ell'-1}\Big)V^{(|v|)}_\lambda\Big(I_u(f_u^{(\bsalpha_u)}\prod_{j\in u\setminus v}\rho_j)\Big)\\
    \leq &2^{|u|-|v|}\Big(\prod_{j\in u\setminus v}\prod_{\ell'\in \kappa_j}2^{-\ell'-1}\Big)V^{v}_\lambda(f^{(\bsalpha_u)}_u)
\end{align*}
because $\prod_{j\in u\setminus v} \rho_j\in D_{v^c}$ from equation~\eqref{eqn:Du}. Our conclusion follows once we put the above bound into equation~\eqref{eqn:fkg} and plug in
\begin{align*}
 \prod_{j\in v} 4^{-\lambda \ell_j}\prod_{\ell'\in \kappa^+_j}4^{-\ell'-1}=4^{|v|}\prod_{j\in v} 4^{(1-\lambda) \lceil\kappa_j\rceil_{\alpha+1}}\prod_{\ell'\in \kappa_j}4^{-\ell'-1}    
\end{align*}
and $\lceil\kappa_j\rceil_{\alpha+1}=0$ for $j\in u\setminus v$.

\section{Proof of Lemma~\ref{lem:fukbound}}\label{app4}
By equation~\eqref{eqn:fuWalsh}, $\hat{f}_u(\bsk)=0$ if $\supp(\bsk)\neq u$. Letting 
$$\natu^u_{\alpha}=\{(k_j,j\in u )\in \natu^{|u|}\mid |\kappa_j|\leq \alpha+1 \text{ for } j\in u\},$$
we can apply Lemma~\ref{lem:fkbound} and equation~\eqref{eqn:funorm} to get
\begin{align}\label{eqn:fukbound}
&\sum_{\bsk\in\natu_*^{s}\setminus K_u(T)} |\hat{f}_u(\bsk)|^2 \\
\leq & 4^{|u|} \Vert f\Vert^2_{u,\alpha,\lambda} \sum_{\bsk_u\in \natu^u_{\alpha}}\bsone\Big\{\bsk_u\notin \bigcup_{v\subseteq u}\mathcal{K}_{u,v}(T) \Big\} \prod_{j\in u}4^{(1-\lambda) \lceil\kappa_j\rceil_{\alpha+1}}\prod_{\ell'\in \kappa_j}4^{-\ell'-1}. \nonumber  
\end{align}
Because $\bsk_u\notin \bigcup_{v\subseteq u}\mathcal{K}_{u,v}(T)$ implies
    $$(\lambda-1)\sum_{j\in u} \lceil\kappa_j\rceil_{\alpha+1}+ \sum_{j\in u} \Vert \kappa_j\Vert >T,$$
    we have for any $\theta\in (0,1)$ and $\rho_\theta=4^{-1+\theta}$
    \begin{align*}
        &\sum_{\bsk_u\in \natu^u_{\alpha}} \bsone\Big\{\bsk_u\notin \bigcup_{v\subseteq u}\mathcal{K}_{u,v}(T) \Big\} \prod_{j\in u}4^{(1-\lambda) \lceil\kappa_j\rceil_{\alpha+1}}\prod_{\ell'\in \kappa_j}4^{-\ell'-1} \\
        \leq & \sum_{\bsk_u\in \natu^u_{\alpha}} 4^{-\theta T+\theta(\lambda-1)\sum_{j\in u} \lceil\kappa_j\rceil_{\alpha+1}+\theta\sum_{j\in u} \Vert \kappa_j\Vert} \prod_{j\in u}4^{(1-\lambda) \lceil\kappa_j\rceil_{\alpha+1}}\prod_{\ell'\in \kappa_j}4^{-\ell'-1} \\
        = & 4^{-\theta T} \sum_{\bsk_u\in \natu^u_{\alpha}} \prod_{j\in u}4^{-(1-\theta)(\lambda-1) \lceil\kappa_j\rceil_{\alpha+1}}\prod_{\ell'\in \kappa_j}4^{-(1-\theta)\ell'-1} \\
        = & 4^{-\theta T} \Big( \sum_{\alpha'=0}^\alpha 4^{-\alpha'} \sum_{\kappa\in\mathcal{S}_{\alpha'}(\natu)}\prod_{\ell'\in \kappa}\rho_\theta^{\ell'} \quad + \quad4^{-\alpha-1}\sum_{\kappa\in \mathcal{S}_{\alpha+1}(\natu)}\rho_\theta^{(\lambda-1) \lceil\kappa\rceil_{\alpha+1}}\prod_{\ell'\in \kappa}\rho_\theta^{\ell'} \Big)^{|u|},
    \end{align*}
    where $\mathcal{S}_{\alpha'}(\natu)=\{ \kappa\subseteq \natu\mid |\kappa|=\alpha'\} $. To bound the above sum, consider the generating function 
    $$g(t)=\Big(\sum_{N=1}^\infty t^N\Big)^{\alpha'},$$
    and let $[t^N]g(t)$ denote the coefficient of $t^N$ in $g(t)$ (see \cite{flaj:sedg:2009} for background knowledge on generating functions). It follows that $[t^N]g(t)$ gives the number of ways to choose an ordered array of length $\alpha'$ from $\natu$ while ensuring the elements in the array sum to $N$.
    Because each $\kappa\in \mathcal{S}_{\alpha'}(\natu)$ is repeated $(\alpha')!$ times in the above counting, we have 
    \begin{equation}\label{eqn:generatingfun}
     \sum_{\kappa\in\mathcal{S}_{\alpha'}(\natu)}\prod_{\ell'\in \kappa}t^{\ell'}\leq \frac{1}{(\alpha')!}\Big(\sum_{N=1}^\infty t^N\Big)^{\alpha'}=\frac{1}{(\alpha')!(t^{-1}-1)^{\alpha'}}. 
    \end{equation}
    with $t=\rho_\theta$. For $\kappa\in \mathcal{S}_{\alpha+1}(\natu)$, we let $\kappa^+=\lceil\kappa\rceil_{1{:}\alpha}$ and $\ell=\lceil\kappa\rceil_{\alpha+1}$. By the one-to-one correspondence between $\kappa^+$ and $\kappa^+-\ell=\{\ell'-\ell\mid \ell'\in \kappa\}\in \mathcal{S}_{\alpha}(\natu)$,
    \begin{align*}
        \sum_{\kappa\in \mathcal{S}_{\alpha+1}(\natu)}\rho_\theta^{(\lambda-1) \lceil\kappa\rceil_{\alpha+1}}\prod_{\ell'\in \kappa}\rho_\theta^{\ell'}
    =&\sum_{\ell=1}^\infty\rho_\theta^{(\alpha+\lambda) \ell}\sum_{\kappa^+-\ell\in\mathcal{S}_{\alpha}(\natu)}\prod_{\ell'\in \kappa^+-\ell}\rho_\theta^{\ell'} \\
    =& \sum_{\ell=1}^\infty\rho_\theta^{(\alpha+\lambda) \ell} \frac{1}{\alpha!(\rho_\theta^{-1}-1)^{\alpha}}\\
    =& \frac{1}{\alpha!(\rho^{-\alpha-\lambda}_\theta-1)(\rho_\theta^{-1}-1)^{\alpha}}.
    \end{align*}
    Our conclusion follows once we put the above bounds into equation~\eqref{eqn:fukbound}.

\section{Proof of Lemma~\ref{lem:KuTbound}}\label{app1}
The conclusion is trivial when $T<0$, so we assume $T\geq 0$. As in the proof Lemma~\ref{lem:fkbound}, for each $\bsk_u\in \mathcal{K}_{u,v}(T)$ we write $\kappa^+_j=\lceil\kappa_j\rceil_{1{:}\alpha}$ and  $\ell_j=\lceil\kappa_j\rceil_{\alpha+1}$ for $j\in v$. By equation~\eqref{eqn:Bku}, each $\bsk'=(k'_1,\dots,k'_s)\in B(\bsk_u,v)$ must have the same $\kappa'_j$ for $j\in 1{:}s\setminus v$. For $j\in v$, however, $\kappa'_j\setminus \kappa_j$ can be any subset of $1{:}(\ell_j-1)$ if $\ell_j\geq 2$ and $\kappa'_j=\kappa_j$ if $\ell_j=1$. Thus, $|B(\bsk_u,v)|=\prod_{j\in v} 2^{\ell_j-1}$. Because $B(\bsk_u,v)$  are mutually disjoint for distinct $ \bsk_u\in \mathcal{K}_{u,v}(T)$, 
\begin{equation}\label{eqn:KuT}
    |K_u(T)|=\sum_{v\subseteq u}\sum_{\bsk_u\in \mathcal{K}_{u,v}(T)}\prod_{j\in v} 2^{\ell_j-1}.
\end{equation}
Next, we notice
$$(\lambda-1)\sum_{j\in u} \lceil\kappa_j\rceil_{\alpha+1}+ \sum_{j\in u}\Vert\kappa_j\Vert=(\alpha+\lambda)\sum_{j\in v}\ell_j+\sum_{j\in v}\sum_{\ell'\in\kappa^+_j}(\ell'-\ell_j)+\sum_{j\in u\setminus v}\Vert\kappa_j\Vert$$
and there is a one-to-one correspondence between $\kappa^+_j$ and $\kappa^+_j-\ell_j=\{\ell'-\ell_j\mid \ell'\in \kappa^+_j\}$ for $j\in v$. Furthermore, $\bsk_u\in \ck_{u,v}(T)$ implies
$$\sum_{j\in v}\sum_{\ell'\in\kappa^+_j}(\ell'-\ell_j)+\sum_{j\in u\setminus v}\Vert\kappa_j\Vert=\sum_{j\in v}\Vert\kappa^+_j-\ell_j\Vert+\sum_{j\in u\setminus v}\Vert\kappa_j\Vert\leq T-(\alpha+\lambda)\Vert \bsell_v\Vert_1.$$
For $\bsalpha_u=(\alpha_j,j\in u)\in \ints_{\le \alpha}^{|u|}$, let
$$\tildeK_{u}(N,\bsalpha_u)=\Big\{(k_j,j\in u )\in \natu_0^{|u|}\Big\vert |\kappa_j|=\alpha_j \ \forall j\in u,\sum_{j\in u}\Vert\kappa_j\Vert= N\Big\}.$$
We can decompose $\mathcal{K}_{u,v}(T)$ into subsets with the same $\bsell_v\in \natu^{|v|}$ and $(\kappa^+_j-\ell_j, j\in v; \kappa_j, j\in u\setminus v) \in \tildeK_u(N,\bsalpha_u)$ for $N\leq T-(\alpha+\lambda)\Vert\bsell_v\Vert_1$, and rewrite equation~\eqref{eqn:KuT} as
\begin{align}\label{eqn:KuT2}
    |K_u(T)|=&\sum_{v\subseteq u}\sum_{\substack{\bsalpha_u\in \ints_{\le \alpha}^{|u|}\\ \alpha_j=\alpha \ \forall j\in v,\\ \alpha_j>0 \ \forall j\in u\setminus v  }}\sum_{\substack{\bsell_v\in \natu^{|v|} } } \sum_{\substack{N\in \natu_0\\ N\leq T-(\alpha+\lambda)\Vert\bsell_v\Vert_1}}|\tildeK_u(N,\bsalpha_u)|\prod_{j\in v} 2^{\ell_j-1}\\
    =&\sum_{\substack{\bsalpha_u\in \ints_{\le \alpha}^{|u|} \\ \alpha_j>0 \ \forall j\in u} }\sum_{ N=0}^{\lfloor T\rfloor}|\tildeK_u(N,\bsalpha_u)|+\nonumber\\
    &\sum_{\substack{v\subseteq u\\ v\neq \emptyset}} 2^{-|v|}\sum_{\substack{\bsalpha_u\in \ints_{\le \alpha}^{|u|}\\ \alpha_j=\alpha \ \forall j\in v,\\ \alpha_j>0 \ \forall j\in u\setminus v  }}\sum_{ N=0}^{\lfloor T\rfloor}|\tildeK_u(N,\bsalpha_u)|\sum_{\substack{\bsell_v\in \natu^{|v|} \\ \Vert\bsell_v\Vert_1\leq (T-N)/(\alpha+\lambda)} }\hspace{-1em}2^{\Vert\bsell_v\Vert_1} \nonumber .
\end{align}
    Let $\overline{N}=\lfloor (T-N)/(\alpha+\lambda)\rfloor$. For nonempty $v\subseteq u$ and integer $n\geq |v|$, there are ${n-1\choose |v|-1}$ number of $\bsell_v\in \natu^{|v|}$ satisfying $\Vert\bsell_v\Vert_1=n$. Hence, when $\overline{N}\geq |v|$,
\begin{align}\label{eqn:sumoverl}
  \sum_{\substack{\bsell_v\in \natu^{|v|} \\ \Vert\bsell_v\Vert_1\leq (T-N)/(\alpha+\lambda)} } \hspace{-1em}2^{\Vert\bsell_v\Vert_1}=&\sum_{n=|v|}^{\overline{N} }{n-1\choose |v|-1} 2^n\leq \sum_{n=|v|}^{\overline{N} }{\overline{N}-1\choose |v|-1} 2^n \leq  {\overline{N}-1\choose |v|-1} 2^{\overline{N}+1}.
\end{align}
When $\overline{N}<|v|$, no $\bsell_v\in \natu^{|v|}$ satisfies $\Vert\bsell_v\Vert_1\leq (T-N)/(\alpha+\lambda)$ and the above sum is trivially $0$. Putting the above bound into equation~\eqref{eqn:KuT2}, we get
\begin{align}\label{eqn:KuT3}
    &|K_u(T)|-\sum_{\substack{\bsalpha_u\in \ints_{\le \alpha}^{|u|} \\ \alpha_j>0 \ \forall j\in u} }\sum_{ N=0}^{\lfloor T\rfloor}|\tildeK_u(N,\bsalpha_u)|\\
    \leq &\sum_{\substack{v\subseteq u\\ v\neq \emptyset}} 2^{-|v|}\sum_{\substack{\bsalpha_u\in \ints_{\le \alpha}^{|u|}\\ \alpha_j=\alpha \ \forall j\in v,\\ \alpha_j>0 \ \forall j\in u\setminus v  }}\sum_{N=0}^{\lfloor T\rfloor}|\tildeK_u(N,\bsalpha_u)|  {\lfloor (T-N)/(\alpha+\lambda)\rfloor \choose |v|-1} 2^{(T-N)/(\alpha+\lambda)+1}\nonumber\\
    \leq &2^{T/(\alpha+\lambda)}\sum_{\substack{v\subseteq u\\ v\neq \emptyset}} 2^{-|v|+1}{\lfloor T/(\alpha+\lambda)\rfloor \choose |v|-1} \sum_{\substack{\bsalpha_u\in \ints_{\le \alpha}^{|u|}\\ \alpha_j=\alpha \ \forall j\in v,\\ \alpha_j>0 \ \forall j\in u\setminus v  }}\sum_{N=0}^{\infty}|\tildeK_u(N,\bsalpha_u)|  2^{-N/(\alpha+\lambda)}.\nonumber
\end{align}
 To bound the sum over $N$, consider the generating function
  $$g(t)=\prod_{j\in u} \Big(\sum_{N=1}^\infty t^N\Big)^{\alpha_j}.$$
  Because each element of $\tildeK_u(N,\bsalpha_u)$ is counted $\prod_{j\in u}(\alpha_j)!$ times by $t^N[g(t)]$ , an argument similar to equation~\eqref{eqn:generatingfun} gives
 \begin{equation}\label{eqn:KutN}
   \sum_{N=0}^{\infty}|\tildeK_u(N,\bsalpha_u)|  t^N\leq \prod_{j\in u} \frac{1}{(\alpha_j)!}\Big(\sum_{N=1}^\infty t^N\Big)^{\alpha_j}=\prod_{j\in u} \frac{1}{(\alpha_j)!(t^{-1}-1)^{\alpha_j}}  
 \end{equation}
 for $t=2^{-1/(\alpha+\lambda)}$ and
 \begin{align*}
     \sum_{\substack{\bsalpha_u\in \ints_{\le \alpha}^{|u|}\\ \alpha_j=\alpha \ \forall j\in v,\\ \alpha_j>0 \ \forall j\in u\setminus v  }}\sum_{N=0}^{\infty}|\tildeK_u(N,\bsalpha_u)|  2^{-N/(\alpha+\lambda)} 
     \leq & \sum_{\substack{\bsalpha_u\in \ints_{\le \alpha}^{|u|}\\ \alpha_j=\alpha \ \forall j\in v,\\ \alpha_j>0 \ \forall j\in u\setminus v  }}\prod_{j\in u} \frac{1}{(\alpha_j)!(2^{1/(\alpha+\lambda)}-1)^{\alpha_j}} \\
     = &A_{\alpha,\lambda}^{|v|}B_{\alpha,\lambda}^{|u|-|v|},
 \end{align*}
 where $A_{\alpha,\lambda},B_{\alpha,\lambda}$ are defined in equation~\eqref{eqn:ABconstant}.
 By a similar calculation,
 \begin{align*}
     \sum_{\substack{\bsalpha_u\in \ints_{\le \alpha}^{|u|} \\ \alpha_j>0\  \forall j\in u} }\sum_{ N=0}^{\lfloor T\rfloor}|\tildeK_u(N,\bsalpha_u)| \leq & 2^{T/(\alpha+\lambda)}\sum_{\substack{\bsalpha_u\in \ints_{\le \alpha}^{|u|} \\ \alpha_j>0\  \forall j\in u} }\sum_{ N=0}^{\infty} |\tildeK_u(N,\bsalpha_u)|2^{-N/(\alpha+\lambda)} \\
     \leq  &  2^{T/(\alpha+\lambda)} B^{|u|}_{\alpha,\lambda}.
 \end{align*}
 Putting the above bounds into equation~\eqref{eqn:KuT3}, we finally get
 \begin{align}\label{eqn:KuT4}
     |K_u(T)| \leq & 2^{T/(\alpha+\lambda)} \Big(B^{|u|}_{\alpha,\lambda}+\sum_{\substack{v\subseteq u\\ v\neq \emptyset}} 2^{-|v|+1}{\lfloor T/(\alpha+\lambda)\rfloor \choose |v|-1}  A_{\alpha,\lambda}^{|v|}B^{|u|-|v|}_{\alpha,\lambda}\Big) \\
     \leq &2^{T/(\alpha+\lambda)} \Big(B^{|u|}_{\alpha,\lambda}+\sum_{\substack{v\subseteq u\\ v\neq \emptyset}}  \max\big(T/(\alpha+\lambda),1\big)^{|v|} A_{\alpha,\lambda}^{|v|}B_{\alpha,\lambda}^{|u|-|v|}\Big)\nonumber\\
     =&2^{T/(\alpha+\lambda)} \Big(B_{\alpha,\lambda}+\max\big(T/(\alpha+\lambda),1\big) A_{\alpha,\lambda}\Big)^{|u|}. \nonumber 
 \end{align}

\section{Proof of Lemma~\ref{lem:KuTbound2}}\label{app2}

 We adopt the notation used in Appendix~\ref{app1}. For $\bsk_u\in \mathcal{K}_{u,v}(T)$, we again consider the one-to-one correspondence between $\kappa^+_j$ and $\kappa^+_j-\ell_j$ for $j\in v$. Because $d\in \natu$ and $d<\alpha+\lambda$, we must have $d\leq \alpha$.
    By $|\kappa^+_j|=\alpha\geq d$  for $j\in v$,
    $$\sum_{j\in v}\Vert\kappa^+_j-\ell_j\Vert_{(d)}=\sum_{j\in v}\Vert\kappa^+_j\Vert_{(d)}-d\ell_j=\Big(\sum_{j\in v}\Vert\kappa_j\Vert_{(d)}\Big)-d\Vert\bsell_v\Vert_1$$
    and 
    $$\sum_{j\in v}\Vert\kappa^+_j-\ell_j\Vert_{(d)}+\sum_{j\in u\setminus v}\Vert\kappa_j\Vert_{(d)}=\Big(\sum_{j\in u}\Vert\kappa_j\Vert_{(d)}\Big)-d\Vert\bsell_v\Vert_1.$$
    Let
    $$\tildeK_{u,d}(N,N',\bsalpha_u)=\Big\{(k_j,j\in u )\in \tildeK_{u}(N,\bsalpha_u)\Big\vert \sum_{j\in u}\Vert\kappa_j\Vert_{(d)}=N'\Big\}.$$
    Similar to equation~\eqref{eqn:KuT2}, we can write
    \begin{align}\label{eqn:KuKu}
       &|K_u(T)\cap K'_u(T')|\\
       =&\sum_{v\subseteq u} \sum_{\substack{\bsalpha_u\in \ints_{\le \alpha}^{|u|}\\ \alpha_j=\alpha \ \forall j\in v,\\ \alpha_j>0 \ \forall j\in u\setminus v  }}\sum_{\substack{\bsell_v\in \natu^{|v|} } } \sum_{\substack{N\in \natu_0,N'\in \natu_0\\T'-d\Vert\bsell_v\Vert_1<N'\leq N\leq T-(\alpha+\lambda)\Vert\bsell_v\Vert_1}}|\tildeK_{u,d}(N,N',\bsalpha_u)|\prod_{j\in v} 2^{\ell_j-1} \nonumber\\
       =& \sum_{\substack{\bsalpha_u\in \ints_{\le \alpha}^{|u|} \\ \alpha_j>0 \ \forall j\in u} }\sum_{\substack{N\in \natu_0,N'\in \natu_0\\T'<N'\leq N\leq T}}|\tildeK_{u,d}(N,N',\bsalpha_u)|+\nonumber\\
       &\sum_{\substack{v\subseteq u\\ v\neq \emptyset}}2^{-|v|} \sum_{\substack{\bsalpha_u\in \ints_{\le \alpha}^{|u|}\\ \alpha_j=\alpha \ \forall j\in v,\\ \alpha_j>0 \ \forall j\in u\setminus v  }} \sum_{\substack{N\in \natu_0,N'\in \natu_0\\N'\leq N\leq T}}|\tildeK_{u,d}(N,N',\bsalpha_u)|\sum_{\substack{\bsell_v\in \natu^{|v|}\\ \Vert\bsell_v\Vert_1>(T'-N')/d, \\ \Vert\bsell_v\Vert_1\leq (T-N)/(\alpha+\lambda) } } 2^{\Vert\bsell_v\Vert_1}.\nonumber
    \end{align}
    The sum over $\bsell_v$ is empty if $(T'-N')/d\geq (T-N)/(\alpha+\lambda) $. Because $T'/d\geq T/(\alpha+\lambda)$ and $N\geq N'$, we further infer the sum is empty unless $N> N_*$ for
    $$N_*=\frac{(\alpha+\lambda)T'-dT}{\alpha+\lambda-d}.$$
    Therefore,
    \begin{align*}
        &\sum_{\substack{N\in \natu_0,N'\in \natu_0\\N'\leq N\leq T}}|\tildeK_{u,d}(N,N',\bsalpha_u)|\sum_{\substack{\bsell_v\in \natu^{|v|}\\ \Vert\bsell_v\Vert_1>(T'-N')/d, \\ \Vert\bsell_v\Vert_1\leq (T-N)/(\alpha+\lambda) } } 2^{\Vert\bsell_v\Vert_1}\\
        \leq & \sum_{\substack{N\in \natu_0\\ N_*<N\leq T}}\sum_{N'=0}^N |\tildeK_{u,d}(N,N',\bsalpha_u)| \sum_{\substack{\bsell_v\in \natu^{|v|}\\ \Vert\bsell_v\Vert_1\leq (T-N)/(\alpha+\lambda) } } 2^{\Vert\bsell_v\Vert_1}\\
        \leq &\sum_{\substack{N\in \natu_0\\ N_*<N\leq T}}|\tildeK_{u}(N,\bsalpha_u)|  {\lfloor (T-N)/(\alpha+\lambda)\rfloor-1 \choose |v|-1} 2^{\lfloor(T-N)/(\alpha+\lambda)\rfloor+1},
    \end{align*}
    where we have applied equation~\eqref{eqn:sumoverl} in the last step.
    Letting $T_*=(T-N_*)/(\alpha+\lambda)$, we put the above bound into equation~\eqref{eqn:KuKu} and get
    \begin{align}\label{eqn:KuKu2}
        &|K_u(T)\cap K'_u(T')|-\sum_{\substack{\bsalpha_u\in \ints_{\le \alpha}^{|u|} \\ \alpha_j>0 \ \forall j\in u}} \sum_{\substack{N\in \natu_0,N'\in \natu_0\\T'<N'\leq N\leq T}}|\tildeK_{u,d}(N,N',\bsalpha_u)|\\
        \leq & \sum_{\substack{v\subseteq u\\ v\neq \emptyset}} 2^{-|v|}\sum_{\substack{\bsalpha_u\in \ints_{\le \alpha}^{|u|}\\ \alpha_j=\alpha \ \forall j\in v,\\ \alpha_j>0 \ \forall j\in u\setminus v  }} \sum_{\substack{N\in \natu_0\\ N_*<N\leq T}}|\tildeK_{u}(N,\bsalpha_u)|  {\lfloor (T-N)/(\alpha+\lambda)\rfloor \choose |v|-1} 2^{(T-N)/(\alpha+\lambda)+1}\nonumber\\
        \leq &2^{T_*}\sum_{\substack{v\subseteq u\\ v\neq \emptyset}} 2^{-|v|+1}{\lfloor T_*\rfloor \choose |v|-1} \sum_{\substack{\bsalpha_u\in \ints_{\le \alpha}^{|u|}\\ \alpha_j=\alpha \ \forall j\in v,\\ \alpha_j>0 \ \forall j\in u\setminus v  }} \sum_{N=0}^{\infty}|\tildeK_{u}(N,\bsalpha_u)|  2^{-N/(\alpha+\lambda)}\nonumber\\
        \leq  &2^{T_*}\Big(B_{\alpha,\lambda}+\max(T_*,1) A_{\alpha,\lambda}\Big)^{|u|},\nonumber
    \end{align}
    where the last step follows from a calculation similar to equation~\eqref{eqn:KuT3} and \eqref{eqn:KuT4}.

    Next, we notice
    \begin{equation}\label{eqn:KuNN}
     \sum_{\substack{N\in \natu_0,N'\in \natu_0\\T'<N'\leq N\leq T}} |\tildeK_{u,d}(N,N',\bsalpha_u)|\leq t^{-T}\sum_{N'=\lceil T'\rceil }^{\lfloor T\rfloor} t^{N'} \sum_{N=N'}^{\lfloor T\rfloor} t^{N-N'} |\tildeK_{u,d}(N,N',\bsalpha_u)|   
    \end{equation}
    for $t=2^{-1/(\alpha+\lambda-d)}$. To bound the sum over $N$, we let $w=\{j\in u\mid \alpha_j>d\}$ for $\bsalpha_u\in \ints_{\leq\alpha}^{|u|}$ and define $\bsalpha^-_u=(\min(\alpha_j,d),j\in u)$ and $\bsalpha^+_u=(\alpha_j-d,j\in w)$ if $w\neq\emptyset$.
    For $\bsk_u\in \tildeK_{u,d}(N,N',\bsalpha_u)$, we further let $\kappa^-_{j,d}=\lceil\kappa_j\rceil_{1{:}\min(\alpha_j,d)}$ for $j\in u$ and $\kappa^+_{j,d}=\kappa_j\setminus \lceil\kappa_j\rceil_{1{:}d}$ for $j\in w$.  Because for $\bsk_u\in \tildeK_{u,d}(N,N',\bsalpha_u)$,
    $$\sum_{j\in w}\Vert\kappa^+_{j,d}\Vert=\sum_{j\in w}\Vert\kappa_{j}\Vert-\sum_{j\in w}\Vert\kappa^-_{j,d}\Vert=\sum_{j\in u}\Vert\kappa_{j}\Vert-\sum_{j\in u}\Vert\kappa_{j}\Vert_{(d)}=N-N',$$
    the number of possible $(\kappa^+_{j,d},j\in w)$ for $\bsk_u\in \tildeK_{u,d}(N,N',\bsalpha_u)$ with a given set of $(\kappa^-_{j,d},j\in u)\in \tildeK_u(N',\bsalpha^-_u)$ is bounded by the number of $(\kappa^+_{j,d},j\in w)$ with $\sum_{j\in w}\Vert\kappa^+_{j,d}\Vert= N-N'$, which is $|\tildeK_w(N-N',\bsalpha^+_u)|$ by definition. This implies 
    $|\tildeK_{u,d}(N,N',\bsalpha_u)|\leq |\tildeK_u(N',\bsalpha^-_u)|\times |\tildeK_w(N-N',\bsalpha^+_u)|$
    and
    \begin{align*}
       \sum_{N=N'}^{\lfloor T\rfloor} t^{N-N'} |\tildeK_{u,d}(N,N',\bsalpha_u)|\leq &|\tildeK_u(N',\bsalpha^-_u)|\sum_{N=N'}^{\lfloor T\rfloor}|\tildeK_w(N-N',\bsalpha^+_u)|t^{N-N'} \\
       \leq &|\tildeK_u(N',\bsalpha^-_u)| \sum_{N=0}^\infty |\tildeK_w(N,\bsalpha^+_u)|t^{N} \\
       \leq &|\tildeK_u(N',\bsalpha^-_u)|\prod_{j\in w}  \frac{1}{(\alpha^+_j)!(t^{-1}-1)^{\alpha^+_j}}, 
    \end{align*}
where we have applied equation~\eqref{eqn:KutN} in the last step. Putting the above bound into equation~\eqref{eqn:KuNN}, we get for any $\theta'\in (0,1)$,
\begin{align*}
& \sum_{\substack{N\in \natu_0,N'\in \natu_0\\T'<N'\leq N\leq T}} |\tildeK_{u,d}(N,N',\bsalpha_u)|\\
\leq &t^{-T}\Big(\prod_{j\in w}  \frac{1}{(\alpha^+_j)!(t^{-1}-1)^{\alpha^+_j}}\Big)\sum_{N'=\lceil T'\rceil }^{\lfloor T\rfloor}  |\tildeK_u(N',\bsalpha^-_u)|t^{N'}\\
\leq & t^{-T+\theta' T'}\Big(\prod_{j\in w}  \frac{1}{(\alpha^+_j)!(t^{-1}-1)^{\alpha^+_j}}\Big)\sum_{N'=0 }^{\infty}  |\tildeK_u(N',\bsalpha^-_u)|t^{(1-\theta')N'}\\
\leq & t^{-T+\theta' T'}\Big(\prod_{j\in w}  \frac{1}{(\alpha^+_j)!(t^{-1}-1)^{\alpha^+_j}}\Big)\Big(\prod_{j\in u} \frac{1}{(\alpha^-_j)!(t^{-1+\theta'}-1)^{\alpha^-_j}}\Big)
\end{align*}
    where we have applied equation~\eqref{eqn:KutN} once more. Our conclusion finally follows from equation~\eqref{eqn:KuKu2}, $t=2^{-1/(\alpha+\lambda-d)}$ and 
    \begin{align*}
        &\sum_{\substack{\bsalpha_u\in \ints_{\le \alpha}^{|u|} \\ \alpha_j>0 \ \forall j\in u}}\Big(\prod_{j\in w}  \frac{1}{(\alpha^+_j)!(t^{-1}-1)^{\alpha^+_j}}\Big)\Big(\prod_{j\in u} \frac{1}{(\alpha^-_j)!(t^{-1+\theta'}-1)^{\alpha^-_j}}\Big)\\
       = & \sum_{w\subseteq u}\Big(\sum_{\alpha'=1}^{\alpha-d}  \frac{1}{(\alpha')!(t^{-1}-1)^{\alpha'}}\Big)^{|w|}\sum_{\substack{\bsalpha^-_u\in \ints^{|u|}_{\leq d}\\\alpha_j=d \ \forall j\in w\\ \alpha^-_j>0 \ \forall j\in u\setminus w } }\Big(\prod_{j\in u} \frac{1}{(\alpha^-_j)!(t^{-1+\theta'}-1)^{\alpha^-_j}}\Big)^{|u|} \\
       =&  \sum_{w\subseteq u}  \Big(\frac{1}{d!(t^{-1+\theta'}-1)^{d}}\sum_{\alpha'=1}^{\alpha-d}  \frac{1}{(\alpha')!(t^{-1}-1)^{\alpha'}}\Big)^{|w|} \Big(\sum_{\alpha'=1}^{d}\frac{1}{(\alpha')!(t^{-1+\theta'}-1)^{\alpha'}}\Big)^{|u|-|w|}\\
       =& D_{\alpha,\lambda,d,\theta'}^{|u|},
    \end{align*}
    where $D_{\alpha,\lambda,d,\theta'}$ is defined in equation~\eqref{eqn:constantD}.

\section*{Acknowledgments}

The author acknowledges the support of the Austrian Science Fund (FWF) 
Project DOI 10.55776/P34808. For open access purposes, the author has 
applied a CC BY public copyright license to any author accepted 
manuscript version arising from this submission.

The author thanks David Krieg and Art Owen for valuable discussions and two anonymous reviewers for many helpful comments.

\bibliographystyle{abbrv} 
% better than abbrvnat that uses author names in the citation text
\bibliography{qmc}

@book{sickel2006smolyak,
  title={{S}molyak's algorithm, sampling on sparse grids and function spaces of dominating mixed smoothness},
  author={Sickel, W. and Ullrich, T.},
  year={2006},
  publisher={Univ.}
}

@article{gnewuch2020explicit,
  title={Explicit error bounds for randomized {Smolyak} algorithms and an application to infinite-dimensional integration},
  author={Gnewuch, M. and Wnuk, M.},
  journal={Journal of Approximation Theory},
  volume={251},
  pages={105342},
  year={2020},
  publisher={Elsevier}
}

@article{pan2025universal,
  title={Universal {$ L_2 $}-approximation using median lattice algorithms},
  author={Pan, Z. and Goda, T. and Kritzer, P.},
  journal={Preprint},
  year={2025}
}

@article{pan2025l_2,
  title={{$ L_2 $}-approximation using median lattice algorithms},
  author={Pan, Z. and Kritzer, P. and Goda, T.},
  journal={Preprint},
  year={2025}
}

@article{krieg2017universal,
  title={A universal algorithm for multivariate integration},
  author={Krieg, D. and Novak, E.},
  journal={Foundations of Computational Mathematics},
  volume={17},
  number={4},
  pages={895--916},
  year={2017},
  publisher={Springer}
}

@article{ullrich2016role,
  title={The role of {Frolov's} cubature formula for functions with bounded mixed derivative},
  author={Ullrich, M. and Ullrich, T.},
  journal={SIAM Journal on Numerical Analysis},
  volume={54},
  number={2},
  pages={969--993},
  year={2016},
  publisher={SIAM}
}

@inproceedings{frolov1976upper,
  title={Upper error bounds for quadrature formulas on function classes},
  author={Frolov, K. K.},
  booktitle={Dokl. Akad. Nauk SSSR},
  volume={231},
  number={4},
  pages={818--821},
  year={1976}
}

@article{goda2018optimal,
  title={Optimal order quadrature error bounds for infinite-dimensional higher-order digital sequences},
  author={Goda, T. and Suzuki, K. and Yoshiki, T.},
  journal={Foundations of Computational Mathematics},
  volume={18},
  number={2},
  pages={433--458},
  year={2018},
  publisher={Springer}
}

@article{goda2016explicit,
  title={An Explicit Construction of Optimal Order Quasi--{Monte Carlo} Rules for Smooth Integrands},
  author={Goda, T. and Suzuki, K. and Yoshiki, T.},
  journal={SIAM Journal on Numerical Analysis},
  volume={54},
  number={4},
  pages={2664--2683},
  year={2016},
  publisher={SIAM}
}

@article{goda2017optimal,
  title={Optimal order quasi-{Monte Carlo} integration in weighted {Sobolev} spaces of arbitrary smoothness},
  author={Goda, T. and Suzuki, K. and Yoshiki, T.},
  journal={IMA Journal of Numerical Analysis},
  volume={37},
  number={1},
  pages={505--518},
  year={2017},
  publisher={Oxford University Press}
}

@article{kunsch2019optimal,
  title={Optimal confidence for {M}onte {C}arlo integration of smooth functions},
  author={Kunsch, R. J. and Rudolf, D.},
  journal={Advances in Computational Mathematics},
  volume={45},
  number={5},
  pages={3095--3122},
  year={2019},
  publisher={Springer}
}

@article{joy1996quasi,
  title={Quasi-{M}onte {C}arlo methods in numerical finance},
  author={Joy, C. and Boyle, P. P. and Tan, K. S.},
  journal={Management science},
  volume={42},
  number={6},
  pages={926--938},
  year={1996},
  publisher={INFORMS}
}

@article{l2009quasi,
  title={Quasi-{M}onte {C}arlo methods with applications in finance},
  author={L'Ecuyer, P.},
  journal={Finance and Stochastics},
  volume={13},
  pages={307--349},
  year={2009},
  publisher={Springer}
}

@article{griebel2017anova,
  title={The {ANOVA} decomposition of a non-smooth function of infinitely many variables can have every term smooth},
  author={Griebel, M. and Kuo, F. Y. and Sloan, I. H.},
  journal={Mathematics of Computation},
  volume={86},
  number={306},
  pages={1855--1876},
  year={2017}
}

@article{NIEDERREITER198851,
title = {Low-discrepancy and low-dispersion sequences},
journal = {Journal of Number Theory},
volume = {30},
number = {1},
pages = {51-70},
year = {1988},
issn = {0022-314X},
doi = {https://doi.org/10.1016/0022-314X(88)90025-X},
url = {https://www.sciencedirect.com/science/article/pii/0022314X8890025X},
author = {Niederreiter, H.}
}

@article{c812a4a5-75ac-3614-936d-5782e44941d0,
 ISSN = {00255718, 10886842},
 author = {Wang, X.},
 journal = {Mathematics of Computation},
 number = {242},
 pages = {823--838},
 publisher = {American Mathematical Society},
 title = {Strong Tractability of Multivariate Integration Using Quasi-{Monte Carlo} Algorithms},
 urldate = {2025-05-15},
 volume = {72},
 year = {2003}
}

@article{LIU2025101935,
title = {Integrability of weak mixed first-order derivatives and convergence rates of scrambled digital nets},
journal = {Journal of Complexity},
volume = {89},
pages = {101935},
year = {2025},
issn = {0885-064X},
doi = {https://doi.org/10.1016/j.jco.2025.101935},
url = {https://www.sciencedirect.com/science/article/pii/S0885064X25000135},
author = {Y. Liu}
}

@article{pan2024automatic,
  title={Automatic optimal-rate convergence of randomized nets using median-of-means},
  author={Pan, Z.},
  journal={Mathematics of Computation},
  year={2025}
}

@article{ye2025medianqmcmethodunbounded,
      title={Median {QMC} method for unbounded integrands over $\mathbb{R}^s$ in unanchored weighted Sobolev spaces}, 
      author={Z. Ye and J. Dick and X. Wang},
      year={2025},
journal={Preprint},
      eprint={2503.05334},
      archivePrefix={arXiv},
      primaryClass={math.NA},
      url={https://arxiv.org/abs/2503.05334}, 
}

@article{chen2025randomintegrationalgorithmhighdimensional,
  title={A random integration algorithm for high-dimensional function spaces},
  author={Chen, L. and Xu, M. and Zhang, H.},
  journal={Mathematics of Computation},
  year={2025}
}

@article{goda2024simpleuniversalalgorithmhighdimensional,
  title={A simple universal algorithm for high-dimensional integration},
  author={Goda, T. and Krieg, D.},
  journal={Numerische Mathematik},
  pages={1--20},
  year={2025},
  publisher={Springer}
}

@article{Goda2024,
	author = {Goda, T. and Suzuki, K. and Matsumoto, M.},
	doi = {10.1137/22M1525077},
	eprint = { https://doi.org/10.1137/22M1525077 },
	journal = {SIAM Journal on Numerical Analysis},
	number = {1},
	pages = {533--566},
	title = {A Universal Median Quasi-Monte Carlo Integration},
	url = {             https://doi.org/10.1137/22M1525077},	
volume = {62},
	year = {2024}
}

@article{SUZUKI20161,
title = {Formulas for the Walsh coefficients of smooth functions and their application to bounds on the Walsh coefficients},
journal = {Journal of Approximation Theory},
volume = {205},
pages = {1-24},
year = {2016},
issn = {0021-9045},
doi = {https://doi.org/10.1016/j.jat.2015.12.002},
url = {https://www.sciencedirect.com/science/article/pii/S0021904515001665},
author = {Suzuki, K. and Yoshiki, T.}
}

@article{dick:2008,
  title={Walsh spaces containing smooth functions and {quasi--Monte Carlo} rules of arbitrary high order},
  author={Dick, J.},
  journal={SIAM Journal on Numerical Analysis},
  volume={46},
  number={3},
  pages={1519--1553},
  year={2008},
  publisher={SIAM}
}

@article{superpolyone,
  title={Super-polynomial accuracy of one dimensional randomized nets using the median-of-means},
  author={Pan, Z. and Owen, A. B.},
  journal = {Mathematics of Computation},
  volume={92},
  number={340},
  pages={805--837},
  year={2023}
}

@article{superpolymulti,
  title={Super-polynomial accuracy of multidimensional randomized nets using the median-of-means},
  author={Pan, Z. and Owen, A. B.},
  journal={Mathematics of Computation},
  volume={93},
  number={349},
  pages={2265--2289},
  year={2024}
}

@article{goda:lecu:2022,
  title={Construction-free median {quasi-Monte Carlo} rules for function spaces with unspecified smoothness and general weights},
  author={Goda, T. and L'Ecuyer, P.},
  journal={SIAM Journal on Scientific Computing},
  volume={44},
  number={4},
  pages={A2765--A2788},
  year={2022},
  publisher={SIAM}
}

@article{kunsch2019solvable,
  title={Solvable integration problems and optimal sample size selection},
  author={Kunsch, R. J. and Novak, E. and Rudolf, D.},
  journal={Journal of Complexity},
  volume={53},
  pages={40--67},
  year={2019}
}

@article{kuo:nuye:2016,
  title={Application of {quasi-Monte Carlo} methods to elliptic {PDEs} with random diffusion coefficients: a survey of analysis and implementation},
  author={Kuo, F. Y. and Nuyens, D.},
  journal={Foundations of Computational Mathematics},
  volume={16},
  number={6},
  pages={1631--1696},
  year={2016},
  publisher={Springer}
}

@article{joe:kuo:2008,
  title={Constructing {Sobol'} sequences with better two-dimensional projections},
  author={Joe, S. and Kuo, F. Y.},
  journal={SIAM Journal on Scientific Computing},
  volume={30},
  number={5},
  pages={2635--2654},
  year={2008},
  publisher={SIAM}
}

@article{bakh:1959,
author = {N. S. Bakhvalov},
year   = 1959,
title  = {On approximate calculation of multiple integrals},
journal= {Vestnik Moskovskogo Universiteta, Seriya Matematiki,
          Mehaniki, Astronomi, Fiziki, Himii},
volume = 4,
pages  = {3--18},
note   = {(In Russian)}
}

@article{cafmowen,
author    = {R. E. Caflisch and W. Morokoff and A. B. Owen},
title     = {Valuation of Mortgage Backed Securities
            using {Brownian} Bridges to Reduce Effective Dimension},
year      = 1997,
journal   = {Journal of Computational Finance},
volume    = 1,
pages     = {27--46}
}

@article{dick:2011,
  title={Higher order scrambled digital nets achieve the optimal rate of the root mean square error for smooth integrands},
  author={Dick, J.},
  journal={The Annals of Statistics},
  volume={39},
  number={3},
  pages={1372--1398},
  year={2011},
  publisher={Institute of Mathematical Statistics}
}

@book{dick:pill:2010,
author = {J. Dick and F. Pillichshammer},
title  = {Digital sequences, discrepancy and quasi-{Monte Carlo} integration},
publisher = {Cambridge University Press},
year    = 2010,
address = {Cambridge}
}

@book{flaj:sedg:2009,
  title={Analytic combinatorics},
  author={Flajolet, P. and Sedgewick, R.},
  year={2009},
  publisher={Cambridge University Press},
  address = {Cambridge}
}

@article{HICKERNELL2003286,
	author = {F. J. Hickernell and H. Niederreiter},
	doi = {10.1016/S0885-064X(02)00026-2},
	issn = {0885-064X},
	journal = {Journal of Complexity},
	note = {Oberwolfach Special Issue},
	number = {3},
	pages = {286–300},
	title = {The existence of good extensible rank-1 lattices},
	url = {https://www.sciencedirect.com/science/article/pii/S0885064X02000262},
	volume = {19},
	year = {2003}
}

@article{kuo:sloa:wasi:wozn:2010,
author   = {F. Y. Kuo and I. H. Sloan and G. W. Wasilkowski and H. Wo\'zniakowski},
title    = {Liberating the dimension},
year     = {2010},
journal  = {Journal of Complexity},
volume   = 26,
pages    = {422--454}
}

@article{meandim,
author  = {R. Liu and A. B. Owen},
year    = 2006,
title   = {Estimating mean dimensionality of Analysis of Variance Decompositions},
journal = {Journal of the American Statistical Association},
number  = 474,
volume  = 101,
pages   = {712--721}
}

@article{mato:1998:2,
title   = {On the {L$^2$}--discrepancy for anchored boxes},
author  = {J. Matou\v{s}ek},
journal = {Journal of Complexity},
volume  = 14,
pages   = {527--556},
year    = 1998
}

@book{nova:wozn:2008,
 title  = {Tractability of multivariate problems},
 volume = {I: standard information for functionals},
 author = {E. Novak and H. Wozniakowski},
 year   = 2008,
 publisher = {European Mathematical Society},
 address = {Zurich}
}

@inproceedings{rtms,
author = {A. B. Owen},
title  = {Randomly Permuted $(t,m,s)$-Nets and $(t,s)$-Sequences},
booktitle ={Monte Carlo and Quasi-Monte Carlo Methods in Scientific Computing},
year   = 1995,
editor = {H.  Niederreiter and P.  J.-S. Shiue},
pages = {299--317},
publisher = {Springer-Verlag},
address = {New York}
}

@article{smol:1963,
author  = {S. A. Smolyak},
year    = 1963,
title   = {Quadrature and interpolation formulas for tensor
products of certain classes of functions},
journal = {Soviet Math. Dokl.},
volume  = 4,
pages   = {240--243}
}

@article{sloa:wozn:1998,
author = {I. H. Sloan and H. Wozniakowski},
title  = {When are quasi-{Monte Carlo} algorithms efficient for
   high dimensional integration?},
journal= {Journal of Complexity},
volume = 14,
pages  = {1--33},
year  = 1998
}

@article{sobol67,
author = {I. M. Sobol'},
title  = {The Distribution of Points in a Cube
          and the Accurate Evaluation of Integrals (in {R}ussian)},
journal= {Zh. Vychisl. Mat. i Mat. Phys.},
year   = 1967,
volume = 7,
pages  = {784--802}
}

\end{document}